\documentclass[12pt]{article}
\usepackage{amsmath}
\usepackage{graphicx,psfrag,epsf}
\usepackage{enumerate}
\usepackage{natbib}
\usepackage{url} 
\usepackage{soul} 

\usepackage[english]{babel}
\usepackage{amsfonts, amssymb, mathrsfs, amsmath, amsthm}
\usepackage{xcolor}
\usepackage{calc}
\usepackage{tikz}
\usepackage{subfig}

\newcommand{\R}{\mathbb{R}}
\newcommand{\N}{\mathbb{N}}
\newcommand{\btheta}{\boldsymbol{\theta}}

\newcommand{\bgamma}{\boldsymbol{\gamma}}
\newcommand{\bepsilon}{\boldsymbol{\varepsilon}}
\newcommand{\by}{\boldsymbol{y}}
\newcommand{\bx}{\boldsymbol{x}}
\newcommand{\bw}{\boldsymbol{w}}

\newcommand{\bI}{\boldsymbol{I}}
\newcommand{\bc}{\boldsymbol{c}}

\newcommand{\vektor}[1]{\begin{pmatrix} #1 \end{pmatrix}}

\DeclareMathOperator{\var}{Var}
\DeclareMathOperator{\tr}{tr}

\theoremstyle{remark}

\newtheorem{lemma}{Lemma}
\newtheorem{proposition}[lemma]{Proposition}

\begin{document}

\def\spacingset#1{\renewcommand{\baselinestretch}%
	{#1}\small\normalsize} \spacingset{1}


\title{\bf Simultaneous inference for misaligned multivariate functional data}
\author{Niels Lundtorp Olsen, Bo Markussen and Lars Lau Raket\\
	Department of Mathematical Sciences, University of Copenhagen}
\maketitle

\begin{abstract}
	We consider inference for misaligned multivariate functional data that represents the same underlying curve, but where the functional samples have systematic differences in shape. In this paper we introduce a class of generally applicable models where warping effects are modeled through nonlinear transformation of latent Gaussian variables and systematic shape differences are modeled by Gaussian processes. To model cross-covariance between sample coordinates we propose a class of low-dimensional cross-covariance structures suitable for modeling multivariate functional data. We present a method for doing maximum-likelihood estimation in the models and apply the method to three data sets. The first data set is from a motion tracking system where the spatial positions of a large number of body-markers are tracked in three-dimensions over time. The second data set consists of longitudinal height and weight measurements for Danish boys. The third data set consists of three-dimensional spatial hand paths from a controlled obstacle-avoidance experiment. We use the developed method to estimate the cross-covariance structure, and use a classification set-up to demonstrate that the method outperforms state-of-the-art methods for handling misaligned curve data.
\end{abstract}

\noindent%
\vfill

\newpage
\spacingset{1.45} 

\section{Introduction}
While the literature and available methods for statistical analysis of univariate functional data have been rapidly increasing during the last two decades, multivariate functional data has been a largely overlooked topic. Extension of univariate methodology to multivariate functional data is often considered a trivial task, but is rarely done in practice. As a result, the non-trivial parts of extending methodology, such as temporal modeling of cross-covariance or warping of misaligned multidimensional signals, have only received little attention.

A wide range of methods for aligning curves are available. For general reviews of the literature on curve alignment, we refer to  \cite{Ramsay, kneip2008combining}, and \cite{wang2015review}. Curve alignment is a nonlinear problem, so for the vast majority of methods, one can not generally expect to align data in a globally optimal way. In the multitude of available methods for univariate functional data, the quality of the results obtained with the available implementations is very variable. Often, good implementations of simple methods outperform far more advanced methods with less polished implementations, even if the advanced methods should be more suitable to the data at hand. From the perspective of multivariate functional data, a major issue is that only very few methods with publicly available implementations support alignment of multivariate curves. 

While misaligned multivariate functional data have been underrepresented in the statistics literature, similar problems have had a central role in other fields. Analysis of misaligned curves in multiple dimensions is fundamental in the shape analysis literature \citep{younes1998computable, sebastian2003aligning, manay2006integral}, where for example closed planar shapes can be thought of as functions $f\,:\,[0,1]\rightarrow \R^2$ with $f(0)=f(1)$. In much shape data, one do not observe the parametrization of these functions, and for closed shapes the start and end points ($0$ and $1$) of the parametrization are arbitrary in terms of the observed data. As an example, consider data consisting of cells outlines obtained from 2D images that have been manually annotated. Here the first annotated point on a cell does not bear any significance---in fact the orientation of the cell is most likely completely random in the image. For this reason, a fundamental direction of theory in the shape analysis literature is built around invariance to parametrization of the function \citep{younes1998computable} as well as other classical shape invariances such as translation, scaling and rotation \citep{kendall1989survey, dryden1998statistical}.

In recent years, the idea of using invariances similar to the shape analysis literature has been introduced as a general tool to analyze functional data \citep{vantini2012definition}. The most notable class of methods are based on elastic distances for functional data analysis  \citep{srivastava2011shape, kurtek2012statistical, tucker2013generative, srivastava_bog}. The fundamental idea underlying these methods is to represent data in terms of  square-root velocity functions and take advantage of the invariance properties of distance on the associated function space, in particular that distances are not affected by warping of the domain in the observed representation. An elastic distance between two curves $f_1$ and $f_2$ can be defined as the minimal distance between the square-root velocity functions associated to $f_1$ and $f_2\circ v$ where the minimum is taken over all possible warps $v$ of $f_2$ (in the original representation). This approach has proven very successful compared to many conventional approaches, and efficient high-quality implementations for various data types and types of analyses are available \citep{FSUsoft, fdasrvf}.

The vast majority of available methods for handling misaligned functional data are heuristic in the sense that they are based on some choice of data similarity measure that is typically not chosen because it fits well with important characteristics of the data. Rather, the typical rationale is computational convenience and/or incremental improvements over other methods. In the shape literature, methods are perhaps less heuristic and more idealistic, in the sense that they are derived from principles of how a distance between shapes should ideally be. This ideal behaviour is typically specified through invariance properties such as the ones described above. In contrast to these approaches for handling misalignment, we propose a full simultaneous statistical model for the fundamental types of variation in misaligned multivariate curves. In particular, we propose to treat amplitude variation and warping variation equally by modeling them as random effects on their respective domains. 

Only few works have previously considered the idea of simultaneously modeling amplitude and warping as random effects. An early example of an integrated statistical model that modeled curve shifts as random Gaussian effects is presented in \cite{ronn}. The simultaneous inference in the model allows data-driven regularization of the magnitude of the shifts through the estimated variance parameters. The idea has been extended to more general warping functions that are modeled by polynomials \citep{Gervini, RonnSkovgaard}, and lately also to include serially correlated noise within the observations of an individual curve \citep{RaketSommerMarkussen}. In addition to the data-driven regularization of the predicted random effects achieved through estimation of variance parameters, the use of likelihood-based inference naturally relate the discrete observation points and the underlying continuous model. This relation avoids many common issues that arise when developing methods for continuous data in the form of pre-smoothed curves. In particular, the pinching problem, where areas with large deviations are compressed by warping to minimize the integrated residual, does not exist for these methods.
Furthermore, the simultaneous modeling of amplitude and warping effects introduces an explicit maximum likelihood criterion for resolving the identifiability problems related to separating warp and amplitude effects \citep{marron2015functional}. The maximum-likelihood estimates induce a separation of the two effects, namely the most likely given the variation observed in the data.

A related class of models with random affine transformations of both warping and amplitude variation have become popular in growth curve analysis \citep{beath2007infant, cole2010sitar}. \cite{hadjipantelis2014analysis, hadjipantelis2015unifying} provide an extension to this in term of a simultaneous mixed-effects model for the scores in separate functional principal component analyses of the amplitude and the warping effects. The simultaneous model allows not only for cross-correlation within the amplitude and warping scores, but also across these two modes of variation. The estimation procedure used in \cite{hadjipantelis2014analysis, hadjipantelis2015unifying}, however, relies on a pre-alignment of the curves that separates the vertical and the horizontal variation.

The major contribution of this paper is a new class of multivariate models that both eliminates the need for pre-smoothing and -alignment of samples and also allows for estimation of cross-correlation between the coordinates of the amplitude effect. In the proposed framework, even if we do not assume any cross-correlation of the amplitude effects, the prediction of warping functions will still take the full multivariate sample into account, and the alignment will thus typically be superior to alignment of the individual coordinates.

\section{Modeling and inference for misaligned multivariate functional data}
\begin{figure}[!th]
	\centering
	\includegraphics[width = \textwidth, trim = 200 80 150 100, clip]{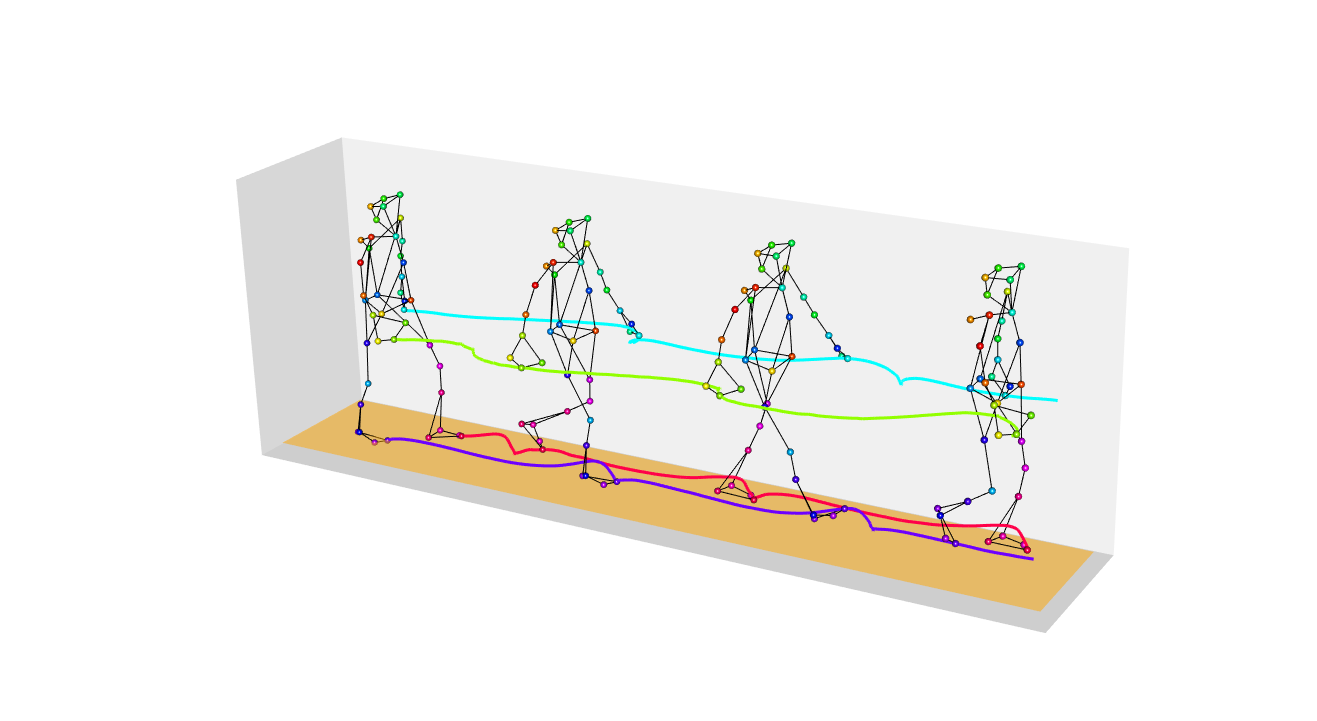}
	\caption{Data from a motion tracking system where the spatial positions of 41 physical markers are tracked in three-dimensions over time. A skeleton model based on the markers is displayed at four temporally equidistant points. The three-dimensional paths of hand and foot markers are displayed.}\label{fig:marker_motivation}
\end{figure}

Consider the multivariate functional observation in Figure~\ref{fig:marker_motivation}. The figure displays a walking sequence in three-dimensional space of a person equipped with 41 markers from the 
\cite{cmu}. The observation is a curve in $\R^{123}$ recorded at 301 time points with a total of 36,963 observed values (20 marker positions missing due to occlusion). 

This sample illustrates some of the challenges in analyzing multivariate functional data. Firstly, a repetition of the walking cycle would in all likelihood produce a trajectory that is visually very similar to the sample, but it would differ in two aspects, the movement timing and the movement path would be slightly different. Such differences in timing and path are random perturbations around the person's ideal walking cycle.  A natural model for such data is thus a nonlinear mixed-effects model where movement timing is modeled as a random effect whose effect is only observed through the nonlinear transformation of the movement path as a function of time, and the movement path variation is modeled as a stochastic process in $\R^{123}$. However, the very large number of observations in a single functional sample puts strong restrictions on the types of models that can be used.  For example, the covariance matrix between the 41 markers at a single time point is $123\times123$, which in practice makes the problem of estimating a single unstructured covariance ($7626$ parameters) impossible. 

\begin{figure}[!t]
	\centering
	\includegraphics[width = 0.65\textwidth]{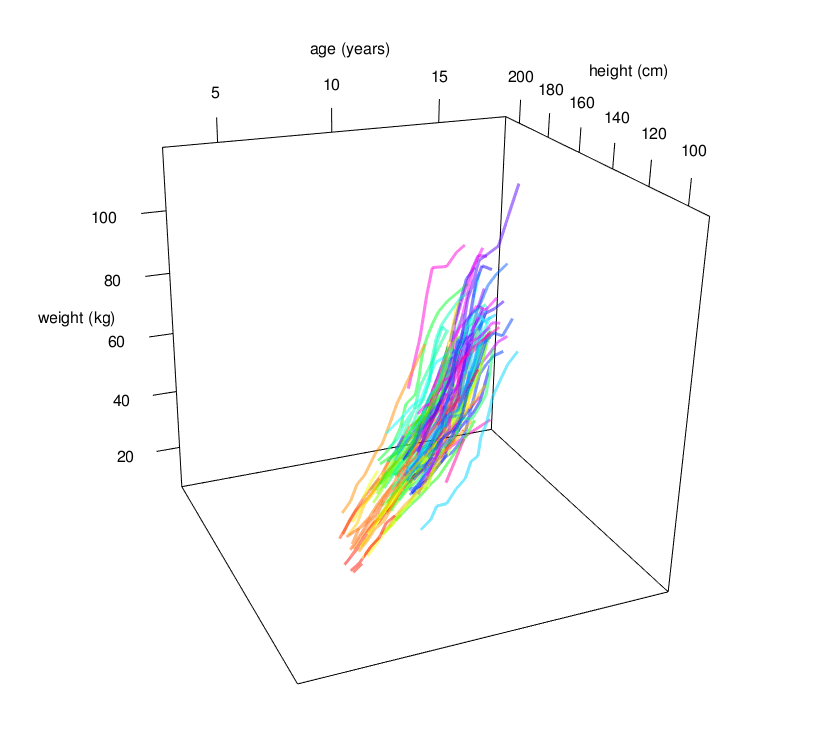}
	
	\caption{Height and weight measurements over time for 106 healthy boys from the Copenhagen Puberty Study. Each individual curve indicates a subject.}
	\label{fig:hw_motivation}
\end{figure}

Another example of multivariate functional data is longitudinal measurements of children's height and weight. Figure~\ref{fig:hw_motivation} displays such data from the Copenhagen Puberty Study \citep{aksglaede2009recent, sorensen2010recent}. The data reflects the fact that height and weight are generally increasing functions during childhood and adolescence. Again, there will be a nonlinear timing effect; observed age is a proxy for a biological or developmental age process of the child, and there will be systematic differences in observation values; taller and heavier children tend to stay taller and heavier than their peers. For height and weight data, one would typically have few observations per child, but the possibility of many children. Thus, the cross-covariance at a given time point could easily be estimated, and one could have a natural interest in inferring possible changes in the correlation between height and weight over time.  

The two above examples illustrate that the challenges of multivariate functional data can be very different. In the following we will introduce a class of models to analyze functional data containing both warp and amplitude variation. To make the model sufficiently flexible, we will introduce generic models for random warping functions and dynamic cross-correlation structures that can approximate arbitrary structures, and whose resolution of approximation can be coarsened by reducing the number of free parameters. 

\subsection{Statistical model}
We consider a set of $N$ discrete observations of $q$-dimensional curves $\by_1,\dots, \by_N\colon [0,1] \rightarrow \R^q$ from $J$ subjects. The curves are assumed to be generated according to the following model
\begin{align}
\by_n(t) = \btheta_{f(n)} (v_n(t)) + \bx_n(t), \quad n=1,\dotsc,N. \label{eq:mod1}
\end{align}
Here $f\colon \{1,\dotsc,N\} \rightarrow \{1,\dotsc,J\}$ is a known function that maps sample number to subject number. The unknown fixed effects are subject specific mean value functions $\btheta_j\colon [0,1] \rightarrow \R^q$ for $j=1,\dotsc,J$ that are modeled using a spline basis assumed to be continuously differentiable. Typical choices are B-spline bases and Fourier bases. The phase variation is modeled by random warping functions $v_n=v(\cdot,\bw_n)\colon [0,1] \rightarrow [0,1]$, which are parametrized by independent latent zero-mean Gaussian variables $\bw_n \in \R^{m_{\bw}}$ for $n=1,\dotsc,N$ with a common covariance matrix $\sigma^2 C$. Here  $v\colon [0,1] \times \R^{m_{\bw}} \to [0,1]$ is a pre-specified function, that is assumed to be continuously differentiable in its second argument, and $m_{\bw} \in \N$ is the dimension of the latent variable.
The amplitude variation is modeled by independent zero-mean Gaussian processes  $\bx_n\colon [0,1] \rightarrow \R^q$ for $n=1,\dotsc,N$ with a common covariance function $\sigma^2 \mathcal{S}$. The unknown variance parameters are thus a scalar $\sigma^2 >0$, a positive definite matrix $C \in \R^{m_{\bw} \times m_{\bw}}$, and a positive definite function $\mathcal{S}\colon [0,1] \times [0,1] \to \R^{q \times q}$. In sections~\ref{sec:warp_models} and \ref{sec:amp_models} we  discuss models for the warping functions and the cross-covariance of the amplitude variation that are highly expressive, while the number of parameters to be estimated is kept at a moderate level. 

We assume that the $n$th curve is observed at $m_n \in \N$ prefixed time points $t_{nk}$, which neither need to be equally spaced in time nor to be shared by the $N$ samples. Stacking the $m_n$ temporally discrete observations into a vector we have
\begin{align}
\vec{\by}_n = \{ \by_n(t_{nk}) + \bepsilon_{nk} \}_{k=1}^{m_n} \in \R^{q m_n}, \quad n=1,\dotsc,N, \label{eq:mod1b}
\end{align}
where the observation noise is given by independent zero-mean Gaussian variables $\bepsilon_{nk} \in \R^q$ with a common variance $\sigma^2 \bI_q$. Here $\bI_q \in \R^{q \times q}$ denotes the identity matrix. 

The major structural difference of model~\eqref{eq:mod1} compared to conventional functional mixed-effects models \citep{Guo} is the inclusion of a warping effect. When compared to conventional methods for curve alignment, the proposed model differs by having a random amplitude effect, by modeling warping functions as random effects, and by handling all effects simultaneously.

\subsection{Modeling warping functions} \label{sec:warp_models}
\begin{figure}[!tp] 
	\centering
	\includegraphics[width = 0.4\textwidth]{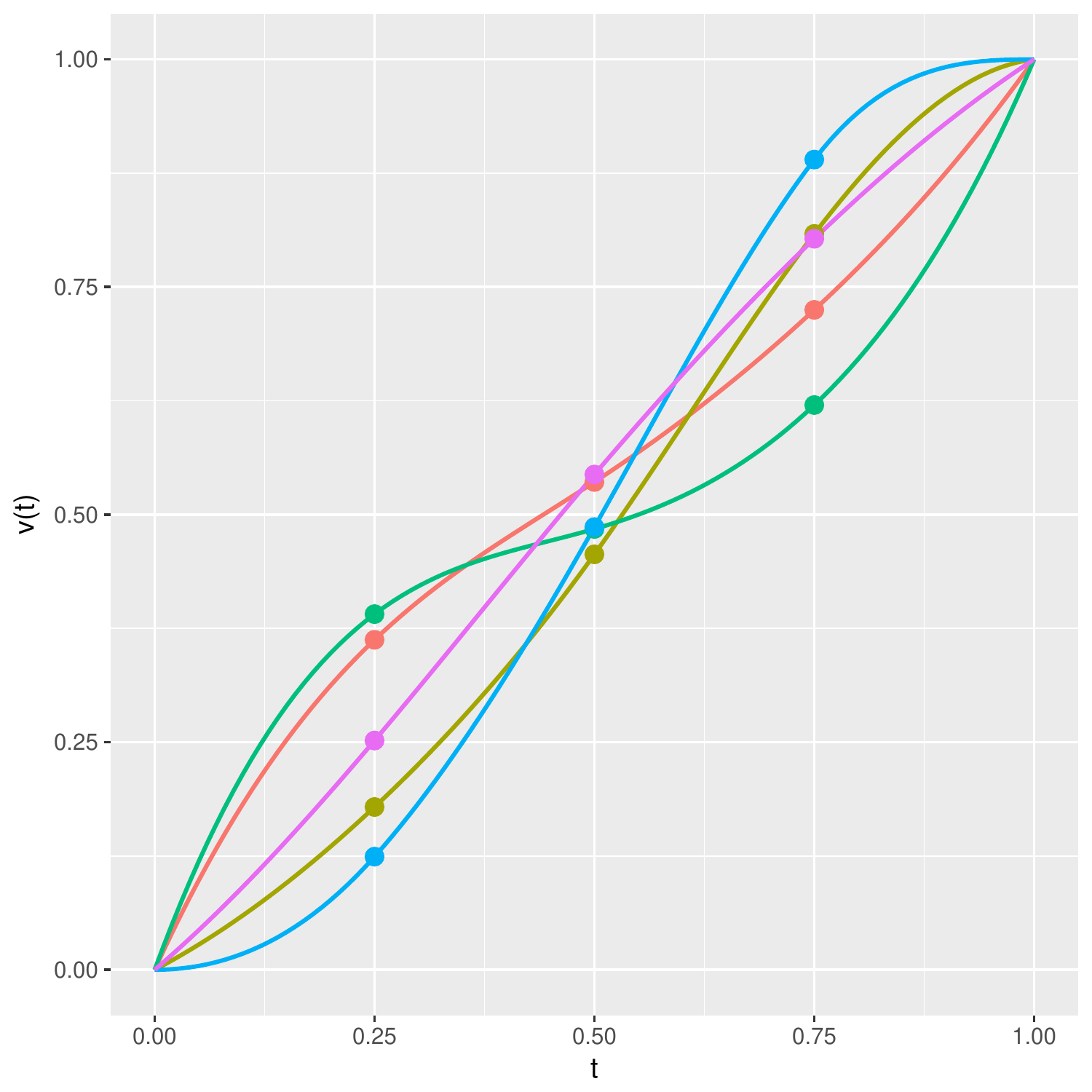}
	\caption{Simulated warping functions with the covariance given by~\eqref{eq:slow_cov}. The warp values at the three interior anchor points are marked by points.} \label{fig:warp_interpolation_slow}
\end{figure}

\begin{figure}[!tp] 
	\centering
	\subfloat[Unit-drift Brownian motion]{\includegraphics[width = 0.4\textwidth]{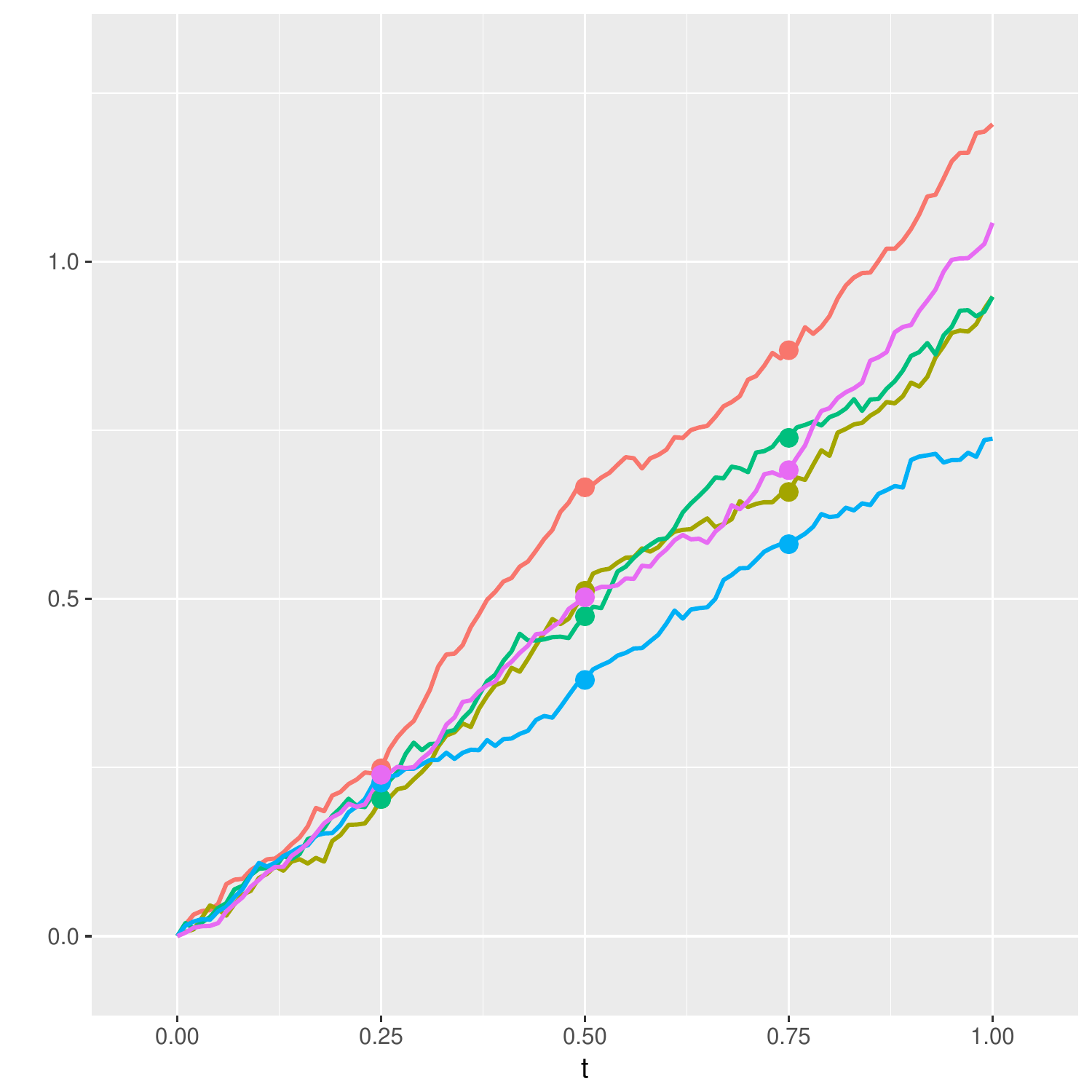}}
	\subfloat[Warping functions corresponding to (a)]{\includegraphics[width = 0.4\textwidth]{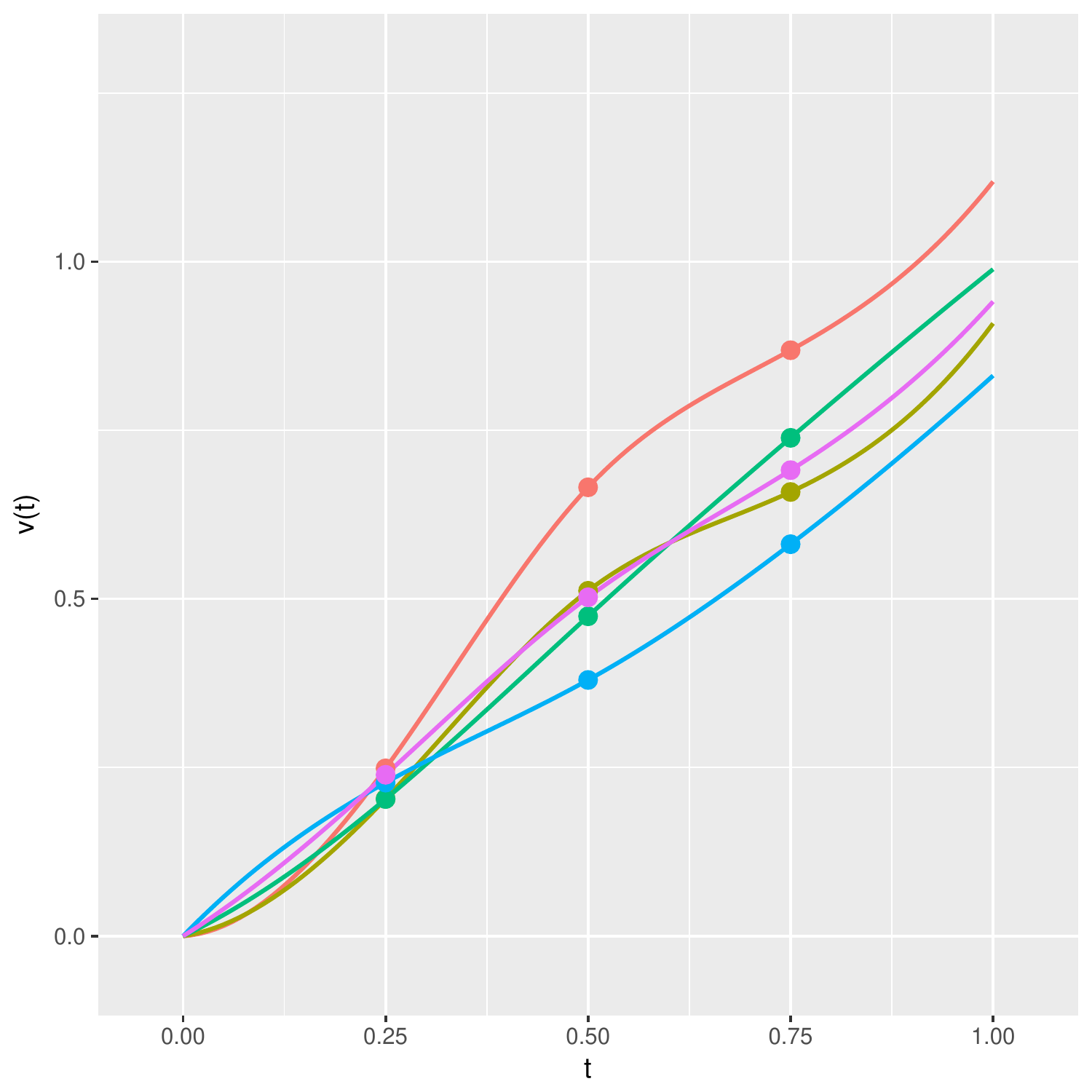}}
	
	\subfloat[Unit-drift Brownian bridge]{\includegraphics[width = 0.4\textwidth]{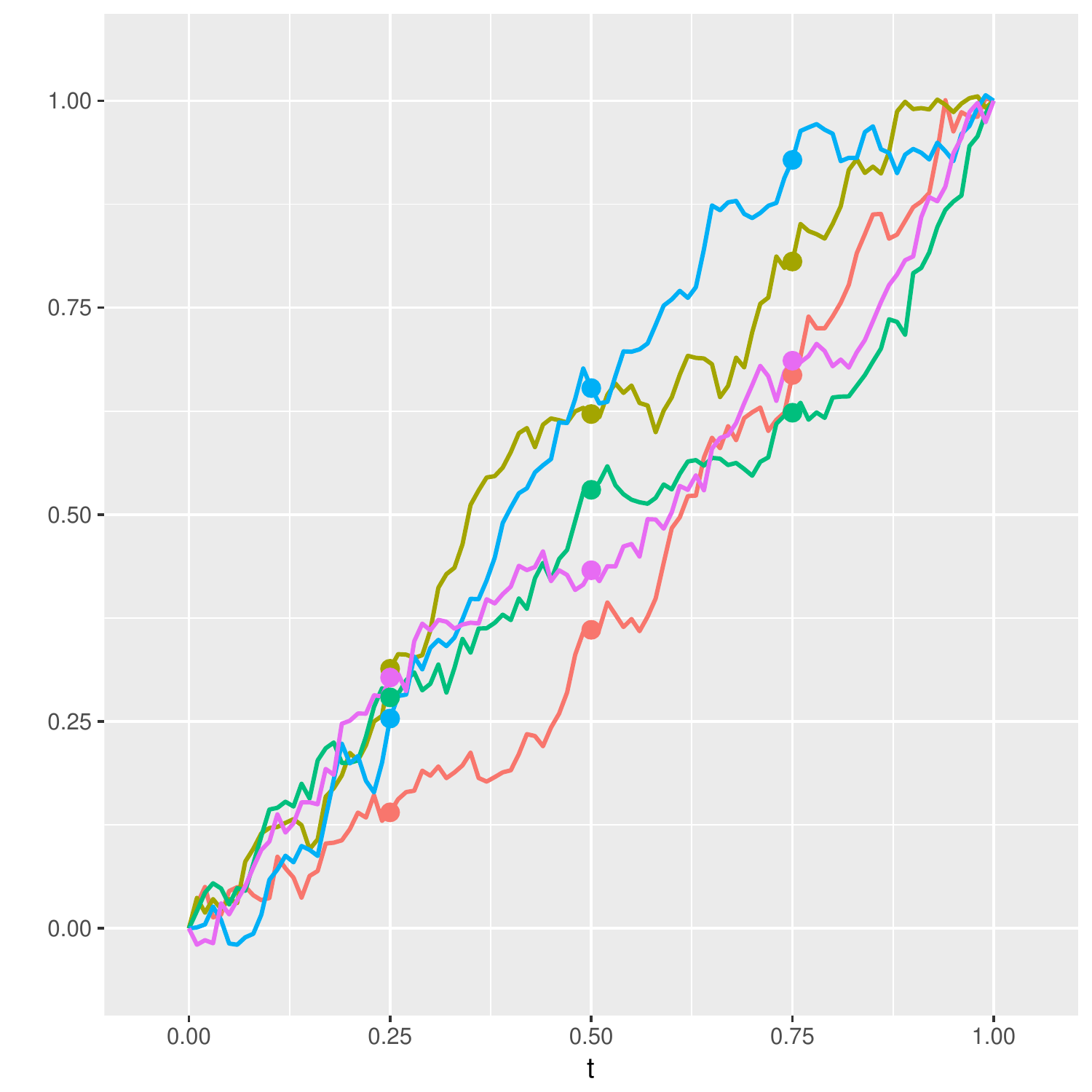}}
	\subfloat[Warping functions corresponding to (c)]{\includegraphics[width = 0.4\textwidth]{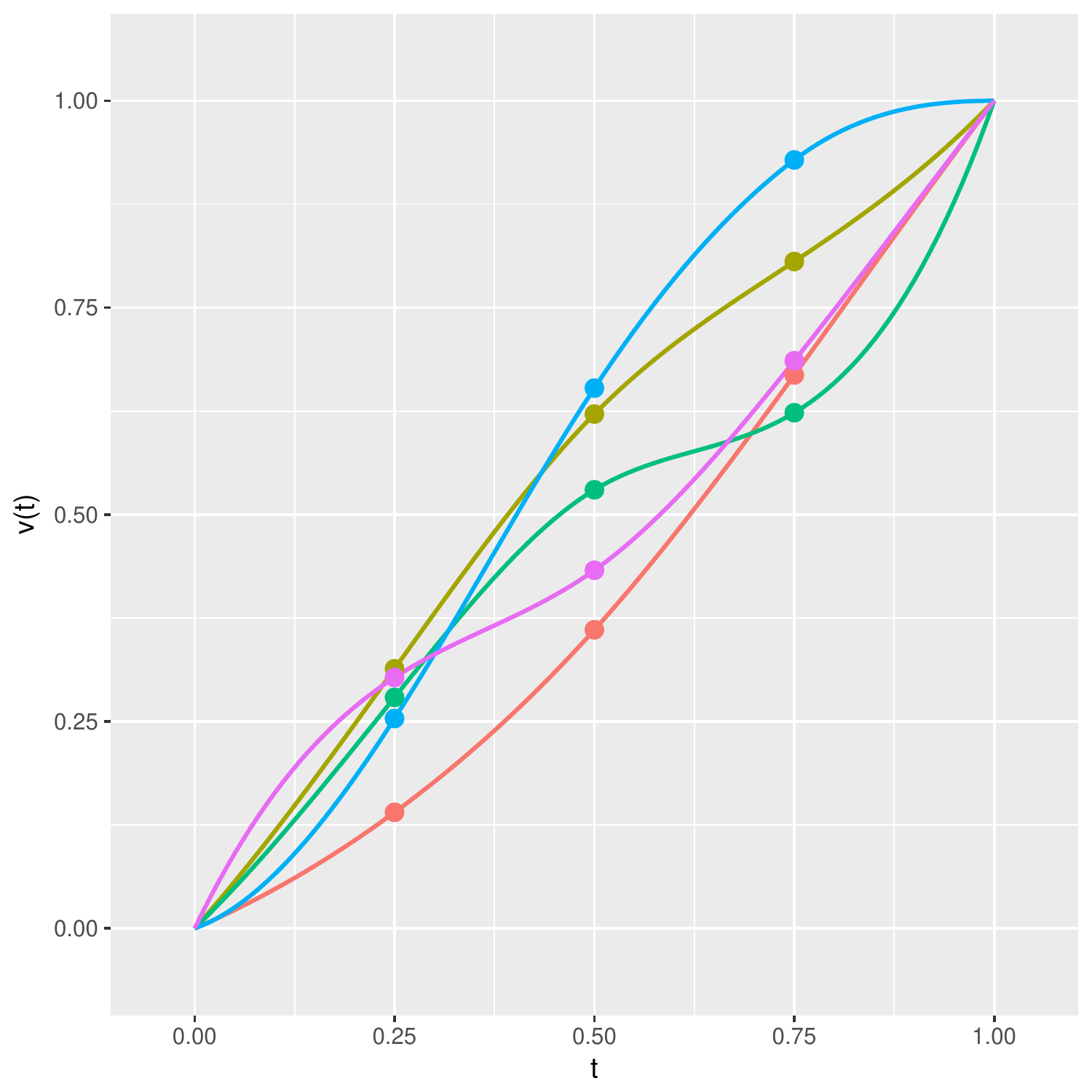}}
	\caption{Constructions of warping functions from stochastic processes with parametric covariances. (a) simulated trajectories of a unit-drift Brownian motion with scale 0.1, (b) warping functions using a unit-drift Brownian motion model with $m_{\bw}=3$ interior equidistant anchor points, fixed interpolation at the left boundary and extrapolation of the rightmost deviation at the right endpoint, (c) simulated trajectories of a unit-drift Brownian bridge with scale  0.2, (d) warping functions using the unit-drift Brownian bridge model with $m_{\bw}=3$ interior equidistant anchor points and fixed interpolation at the boundary.} \label{fig:warp_interpolation}
\end{figure}
The success of the model relies on its ability to approximate the realizations of the true warping functions. To accomplish this, the warping functions $v_n$ must be sufficiently versatile and able to approximate a large array of different warps. We achieve this by modeling warping functions as the identity mapping plus a deformation modeled by interpolating latent warp variables $\bw_n \in \R^{m_{\bw}}$ at pre-specified (e.g. equidistant) anchor points $t_k$ for $k=1,\dotsc,m_{\bw}$
\begin{align}
v_n(t) = v(t, \bw_n) = t + \mathcal{E}_{\bw_n}(t),\label{eq:warpfct}
\end{align}
where the interpolation function $\mathcal{E}_{\bw}$ can, for example, be a linear or a cubic spline. 

The behavior of the predicted warping functions will be determined by the combination of interpolation method (and corresponding boundary conditions) and the estimated covariance of the latent variables $\bw_n$. Throughout this paper we will use cubic spline interpolation of the latent variables. If we think of the parametrization of the $n$th sample, $v_n(t)$, as the internal time of the sample, it is often natural to assume that the internal time is always moving forward. To ensure this, we will predict the latent variables $\bw_n$ using constrained optimization such that the sequence will be increasing along the corresponding anchor points. But for cubic interpolation, a sequence of increasing values at the interpolation points is not sufficient to ensure a monotone interpolation function. To force increasing warping functions we will use the Hyman filter \citep{hyman1983accurate} to ensure that the entire warping function is increasing. For some types of data, it may be meaningful to have warps that can go backwards in time, or it may be useful to include this option to account for uncertainty in the model if the observed signals contain features where the matching is highly ambiguous. Such types of warp models will not be considered in this paper.

The covariance matrix of the latent variables will determine the regularity of the predicted warping functions. When the number of latent variables $m_{\bw}$ is small compared to the number of functional samples $N$ and the number of sampling points $m_1, \dots, m_N$ for the functional samples, one can assume an unstructured covariance and estimate the corresponding $(m_{\bw}^2 + m_{\bw}) / 2$ variance parameters. If the structure of the warping functions are of key interest, one may be able to study the underlying mechanism by estimating an unstructured covariance matrix. Consider for example the simulated warping functions shown in Figure~\ref{fig:warp_interpolation_slow}. These warping functions use the increasing cubic spline construction detailed above with $m_{\bw}=3$ interior equidistant anchor points, fixed boundary points and covariance matrix 
\begin{align}
\begin{pmatrix}0.005 & 0 & -0.004\\ 0 & 0.001 & 0\\ -0.004 & 0 & 0.005\end{pmatrix}.\label{eq:slow_cov}
\end{align}
The interpretation of the strong negative covariance between first and third anchor point suggest a burnout type of process where samples that are ahead initially slow down toward the end and vice versa. The low variance of the middle anchor point suggest that the individual samples are largely synchronized around this time. 

In many cases, one can choose a specific interpolation method and specify a reasonable parametric covariance for the latent variables based on properties of the data.  It is, for example, often natural to think of warping processes as accumulations of small errors causing desynchronization of the set of observed trajectories that all started in the same state. Thinking of Gaussian processes, Brownian motion with linear unit drift would offer a simple model for phenomena where errors are accumulating and increasing the desynchronization of samples over time. Simulations of unit-drift Brownian motions are shown in Figure~\ref{fig:warp_interpolation} (a) and the corresponding simulations of warping functions from $m_{\bw}=3$ interior equidistant anchor points, fixed left boundary point and linear extrapolation of the deviation of the rightmost anchor point at the right boundary point are shown in Figure~\ref{fig:warp_interpolation} (b). 

Suppose we are analyzing longitudinal data of children's heights where we could think of the warping function as the developmental (height) age of the child. At conception (approximately $-9$ months of age), where the child is merely a fertilized egg, all children are the size of a grain of sand and their developmental ages are synchronized. As the children become older the desynchronization of their developmental ages increases. This can, for example, be seen by the vast variation between the age of onset of puberty. The unit-drift Brownian motion warp model seems like a very suitable model for this desynchronization. 

Other types of data may give rise to other models. Consider an experiment that records repetitions of a walking sequence such as the data in Figure~\ref{fig:marker_motivation}, and assume that all sequences start from the same pose and end after two completed gait cycles.
For such data, the desynchronization is not increasing over time since beginning and end poses are synchronized, but we would expect maximum desynchronization around the middle of the gait cycle window. In this setting, a more suitable model would be a unit-drift Brownian bridge as illustrated in Figure~\ref{fig:warp_interpolation} (c) and (d).

	Like other hyperparameters, the number of anchor points is a choice of modelling. However, a low number of anchor points (e.g. 3-5) will generate a class of warp functions that is sufficiently flexible for many applications; we used $m_w = 3$ in all applications presented in this paper. If, however, local variation is very strong and complex and the observed functional samples carry sufficiently clear information about the systematic shapes to recover such complex warps, a higher number of anchor points should be used.

\subsection{Dynamic covariance structures} \label{sec:amp_models}
In the previous section we modeled the covariance structure of smooth warping functions and saw how one could use domain-specific knowledge of the data to choose models with few parameters. Even though the nature of the additive amplitude variation components $\bx_n$ from model~\eqref{eq:mod1} is different, we can extend these ideas to construct parametric, low-dimensional cross-covariance structures that are sufficiently expressive to model a wide array of cross-covariance structures over time. 

\begin{proposition}  \label{prop-dyn-cov}
	Let $f: [0,1] \times [0,1] \rightarrow \R_+$ be a positive definite function on the temporal domain $[0,1]$. Let $0= t_1 < \dots < t_{\ell} = 1$ be anchor points, let $A_1, \dots, A_\ell \in \R^{q \times q}$ be a set of symmetric positive definite matrices, and for each $t \in [0,1]$ define $B_t \in \R^{q \times q}$ as the unique positive definite matrix satisfying 
	\begin{align}
	B_t^\top B_t =   \frac{t_{k+1} - t}{t_{k+1} -t_k} A_k  + \frac{t - t_k}{t_{k+1} -t_k} A_{k+1}  \quad \text{ for } t \in [t_{k}, t_{k+1}].
	\label{eq:dyn_cov}
	\end{align}
	For all $s,t \in [0,1]$, define $K(s,t) = f(s,t) B_s^\top B_t \in \R^{q \times q}$. Then the function $K: [0,1] \times [0,1] \rightarrow \R^{q \times q}$ is positive definite.
\end{proposition}

\begin{proof}
	First we remark that since the space of positive definite  matrices is a convex cone, the linear interpolation $B_t^\top B_t$ is also positive definite, and we may take $B_t$ as the positive square root. 
	To prove that $K: [0,1] \times [0,1] \rightarrow \R^{q \times q}$ is positive definite it suffices to show that the associated finite dimensional marginal matrices are positive definite. Thus, given $s_1,\dotsc,s_m\in [0,1]$ we let the block matrix $\mathbf{V} \in \R^{qm \times qm}$ be defined by
	\begin{equation}
	\mathbf{V} = \vektor{B_{s_1}^\top f(s_1, s_1) B_{s_1} &  B_{s_1}^\top f(s_1, s_2) B_{s_2} & \cdots &  B_{s_1}^\top f(s_1, s_m) B_{s_m} \\ 
		B_{s_2}^\top f(s_2, s_1) B_{s_1} & B_{s_2}^\top f(s_2, s_2) B_{s_2} & \dots &  \\
		\vdots & & \ddots & \\
		B_{s_m}^\top f(s_m, s_1) B_{s_1} & B_{s_m}^\top f(s_m, s_2) B_{s_2} & \dots & B_{s_m}^\top f(s_m, s_m) B_{s_m}}.
	\end{equation}
	By straightforward calculations we have $\mathbf{V} = 
	\mathbf{B}^\top (\mathbf{F} \otimes \mathbf{I}_q) \mathbf{B}$, where $\mathbf{B} \in \R^{qm \times qm}$ is the block-diagonal matrix of $\{B_{s_1}, \dots , B_{s_m} \}$ and 
	\begin{equation}
	\mathbf{F} =  \vektor{f(s_1, s_1) & \cdots & f(s_1, s_m) \\
		\vdots & \ddots & \vdots \\ f(s_m, s_1) & \cdots & f(s_m, s_m)}.
	\end{equation}
	For $z \in \R^{qm} \setminus \{0\}$ we must show that $z^\top \mathbf{V} z > 0$. Setting $u = \mathbf{B} z \neq 0$ and using that $\mathbf{F}$ is positive definite by assumption we have $z^\top \mathbf{V} z = u^\top (\mathbf{F} \otimes \mathbf{I}_q) u > 0$.
\end{proof}
The above proposition gives a general framework for constructing dynamical covariance functions, and it is simple to construct parametric models that allow for estimation of time-varying cross-correlations in a statistical setting. In the statement of the proposition we assumed a common marginal covariance function $f$ along all coordinates. The idea of modeling a cross-covariance structure by linearly interpolating cross-covariances at specific points seamlessly extends to multivariate diagonal covariance functions (i.e. no cross-covariances), such that the individual coordinates of the functional samples may be modeled using different types covariance functions or different parameters.

\section{Estimation} \label{sec:estimation}

Direct likelihood inference in the model \eqref{eq:mod1} is not feasible as the model contains nonlinear latent variables in combination with possible very large data sizes. Instead we propose a maximum-likelihood estimation procedure based on iterative local linearization \citep{LindstromBates}. The procedure is a multivariate extension of the estimation procedure described in \cite{RaketSommerMarkussen}, however with an improved estimation of fixed effects.

The estimation procedure consists of alternating steps of (1); estimating fixed effects (i.e. spline coefficents) and predicting the most likely warp variables given the data and current parameter estimates, (2); estimating variance parameters from the locally linearized likelihood function around the maximum a posteriori predictions  $\bw_1^0,\dots,\bw_N^0$ of the warp variables. The linearization in the latent Gaussian warp parameters $\bw_1,\dots, \bw_N$ means that we approximate the nonlinearly transformed probability density by the density of a linear combination of multivariate Gaussian variables. The estimation procedure is thus a Laplace approximation of the likelihood, and the quality of the approximation is approximately second order \citep{wolfinger1993laplace}.

\paragraph{Predicting warps}
In the first step of the estimation procedure we want to predict the most likely warps from model~\eqref{eq:mod1} given the current parameter estimates. The negative log posterior for a single functional sample is proportional to
\begin{equation}
(\vec{\bgamma}_{\bw_n} - \vec{\by}_n)^\top (\bI_{qm_n} + S_n)^{-1}(\vec{\bgamma}_{\bw_n} - \vec{\by}_n) + \bw_n^\top C^{-1}\bw_n \label{post-lik}
\end{equation}
where $\vec{\bgamma}_{\bw_n} \in \R^{qm_n}$ is the stacked vector $\{ \btheta_{f(n)}(v(t_{nk}, \bw_n)) \}_{k=1}^{m_n}$ and $S_n \in \R^{qm_n \times qm_n}$ is the amplitude covariance $\{ \mathcal{S}(t_{nj},t_{nk}) \}_{j,k=1,\dotsc,m_n}$ at the sample points. The issue of predicting warps is thus a nonlinear least squares problem that can be solved by conventional methods.

\paragraph{Estimating variance parameters}
Since $\btheta_{f(n)}\circ v(t_{nk},\cdot)$ are smooth functions for all $n=1,\dotsc,N$, $k=1,\dotsc,m_n$ we can linearize  model~\eqref{eq:mod1} around a given prediction $\bw_n^0$ using the first-order Taylor expansion. The linearization is given by
\begin{equation}
\btheta_{f(n)}(v(t_{nk} , \bw_n)) \approx \btheta_{f(n)}(v(t_{nk}, \bw_n^0)) +  \partial_t \btheta_{f(n)}(v(t_{nk}, \bw_n^0))(\nabla_w v(t_{nk}, \bw_n^0))^\top (\bw_n - \bw_n^0).
\end{equation}

For the discrete observation of the $n$th curve this gives a linearization of model~\eqref{eq:mod1} as a vectorized linear mixed-effects model on the form 
\begin{equation} \label{lin-mod}
\vec{\by}_n \approx \vec{\bgamma}_{\bw_n^0} +  Z_n(\bw_n - \bw_n^0) + \vec{\bx}_n + \vec{\bepsilon}_n, \quad
n=1,\dotsc,N, 
\end{equation}
where $\vec{\bgamma}_{\bw_n^0}, \vec{\bx}_n, \vec{\bepsilon}_n \in \R^{qm_n}$ are the stacked vectors
\begin{align*}
\vec{\bgamma}_{\bw_n^0} &= \{ \btheta_{f(n)}(v(t_{nk}, \bw_n^0)) \}_{k=1}^{m_n}, &
\vec{\bx}_n &= \{ \bx_n(t_{nk}) \}_{k=1}^{m_n}, &
\vec{\bepsilon}_n &= \{ \bepsilon_{nk} \}_{k=1}^{m_n},
\end{align*}
and $Z_n \in \R^{qm_n \times m_{\bw}}$ is the row-wise stacked matrix
\begin{equation*}
Z_n = \{ \partial_t  \btheta_{f(n)}(v(t_{nk}, \bw_n^0)) \nabla_{\bw} v(t_{nk}, \bw_n^0) \}_{k=1}^{m_n}.
\end{equation*}

In the approximative model~\eqref{lin-mod} twice the negative profile log-likelihood $l(\sigma^2, C, \mathcal{S})$ for the variance parameters is given by
\begin{equation} \label{linearized-like}
\sum_{n=1}^N \bigg( q m_n \log \sigma^2 + \log\det V_n \\
+ \sigma^{-2}(\vec{\by}_n -  \vec{\bgamma}_{\bw_n^0} + Z_n \bw_n^0)^\top V_n^{-1} (\vec{\by}_n -  \vec{\bgamma}_{\bw_n^0} +  Z_n \bw_n^0 ) \bigg),
\end{equation}
where $V_n = Z_n C Z_n^\top + S_n + \bI_{qm_n}$ with $S_n = \{ \mathcal{S}(t_{nj},t_{nk}) \}_{j,k=1,\dotsc,m_n}$. In particular, the profile maximum-likelihood estimate for $\sigma^2$ is given by 
\begin{equation*}
\hat\sigma^2 = \frac{1}{q m}\sum_{n=1}^N (\vec{\by}_n -  \vec{\bgamma}_{\bw_n^0} +   Z_n \bw_n^0 )^\top V_n^{-1} (\vec{\by}_n -  \vec{\bgamma}_{\bw_n^0} +  Z_n \bw_n^0 )
\end{equation*}
where $m = \sum_{n=1}^N m_n$ is the total number of observations. Estimation of the variance parameters $C$ and $\mathcal{S}$ related to the warping and amplitude effects is done using the profile likelihood $l(\hat\sigma^2, C, \mathcal{S})$.

\paragraph{Estimating fixed effects}
As the fixed effects are given by  spline bases, estimation of these can be handled within the framework of linear Gaussian models, remembering that basis functions should be evaluated at warped time points $v_n(t_{nk})$. Since $v_n(t_{nk}) = v(t_{nk},\bw_n)$ is not fixed up front, we are required to recalculate the spline basis matrix for each new prediction of $\bw_n$. This estimation improves that of \cite{RaketSommerMarkussen}, which used a point-wise estimation based on the inverse warp that ignored the amplitude variance of the curves.

There is no closed-form expression for the maximum-likelihood estimator of the fixed effects in the linearized model, since spline coefficients also enter the variance terms through the matrices $Z_n$, as can be seen in equation \eqref{linearized-like}. However, by construction $Z_n$ is linear in the spline coefficients so estimation can be done using an EM algorithm. The details of these calculations can be found in the supplementary material. 

In practice, the estimation in the linearized model can be approximated by estimating from the posterior likelihood \eqref{post-lik} which gives a computationally efficient closed-form solution. The  difference between these two approaches is that the EM algorithm takes the uncertainty in prediction of $\bw_n$ into account and is guaranteed to decrease the linearized likelihood \eqref{linearized-like}. However, for a moderate number of warp parameters, there should only be a small conditional variance on $\bw_n$. 

In the data applications presented in the following sections, we estimated fixed effects from the posterior likelihood. In the last application on hand movements, these posterior likelihood estimates were used to initialized the likelihood optimization which were subsequently fine-tuned by the EM algorithm with a single update per warp prediction. This was done to evaluate if improved likelihood estimates could be obtained, but the EM algorithm offered only a very slight improvement in linearized likelihood. 

\section{Applications}

\begin{figure}[!htp]
	\centering
	\includegraphics[width = 0.55\textwidth, trim = 110 130 120 200, clip]{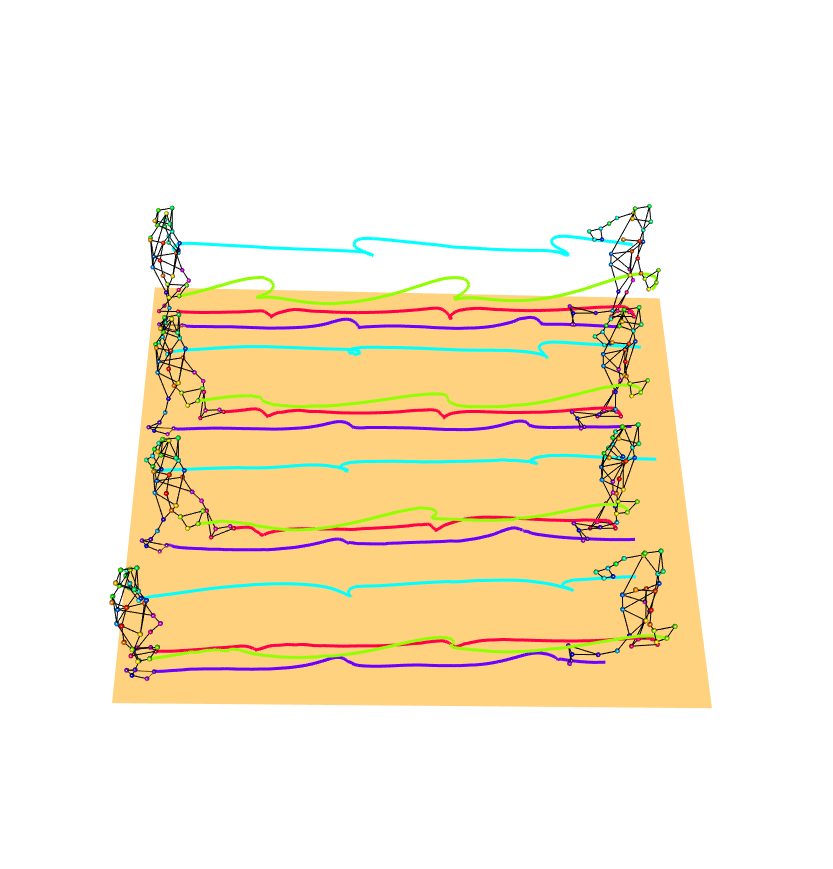}
	\includegraphics[width = 0.55\textwidth, trim = 130 130 120 200, clip]{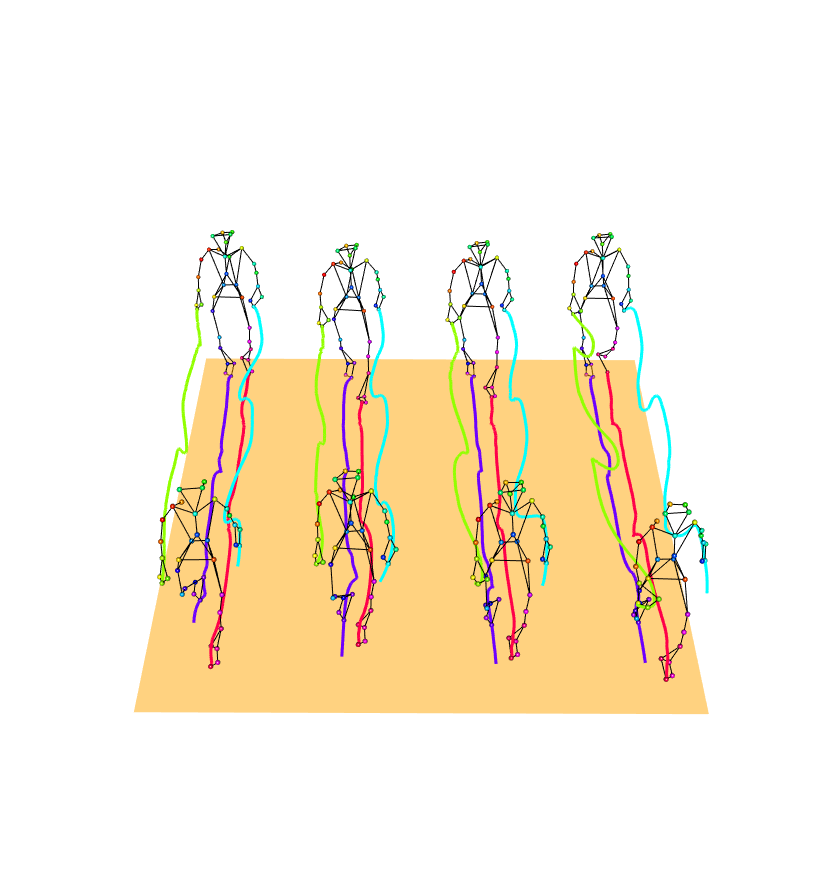}
	
	\caption{Side and frontal view of the motion trajectories of four walking sequences performed by the same participant.}\label{fig:mocap1}
	
\end{figure}
\subsection{Motion capture data}

\paragraph{Data and model} Data consists of four 12-dimensional functional objects. The curves consist of a total of 1284 temporal observations in $\R^{12}$. 
As can be seen in Figure~\ref{fig:mocap1}, the trajectories start and end at different places during the gait cycle. To handle this structure, time was scaled to the interval $[0, 1]$ such that all samples began at $0.1$, and such that the temporally longest trajectory ended at $0.9$. We included random shift parameters $s_n$ in our warping functions to model these different temporal onsets of the gait cycle. The shifts $s_n$ were modeled as Gaussian random variables. The full model is
\begin{equation}
\by_n(t) = \btheta(v(t,\bw_n, s_n)) + \bx_n(t)
\end{equation}
where $\btheta: [0,1] \rightarrow \R^{12}$ is the mean curve for the observations (modeled using a 3-dimensional B-spline basis with 30 interior anchor points) and the warping function $v$ is given by
\begin{equation*}
v(t,\bw_n, s_n) = t + s_n + \mathcal{E}_{\bw_n}(t)
\end{equation*}
where $\mathcal{E}_{\bw_n}$ is an increasing cubic spline interpolation (Hyman filtered) of $\bw_n$ at $m_{\bw}=3$ equidistant anchor points.
No subject-specific effects were included as all responses were recorded from the same individual. The amplitude effect $\bx_n$ was modeled as a Gaussian process with a Mat\'ern covariance $f_\text{Mat\'{e}rn(2,$\kappa$)}(s,t)$ with second order smoothness, assuming independent coordinates and a common range parameter $\kappa$ (see equation~\eqref{cov:matern} in the supplement). We assumed different scaling parameters for each of the 12 coordinates of $\bx_n$. Since the data is roughly cut to include two gait cycles, one would expect high synchronization of start and end poses in percentual time when corrected for the different onsets. Therefore, latent variables $\bw_n$ were modeled as discretely observed Brownian bridges with a single scale parameter. 

\paragraph{Results}
The predicted warping functions are shown in Figure~\ref{fig:mcd-warp}, and the corresponding aligned samples are shown in Figure~\ref{fig:mcd-kurver}. The samples are nicely aligned, in particular, the regular elevation profiles of the left and right feet seems very well aligned. The remaining signals have their key-features aligned, with the residual variation evenly spread out across the coordinates. This is a feature of the simultaneous multivariate fitting, where the best alignment given the variation in the different coordinates is found. Individual alignment of the coordinates would produce warping functions that overfitted the individual aspects of the movement.  
In Figure~\ref{fig:mcd-med-estimater}, we have displayed the estimated mean trajectories $\btheta$ and illustrated the uncertainty after alignment by 95\% prediction ellipsoids for the amplitude effect $\bx_n$.

\begin{figure}[!tp]
	\centering
	\includegraphics[width = 0.8\textwidth]{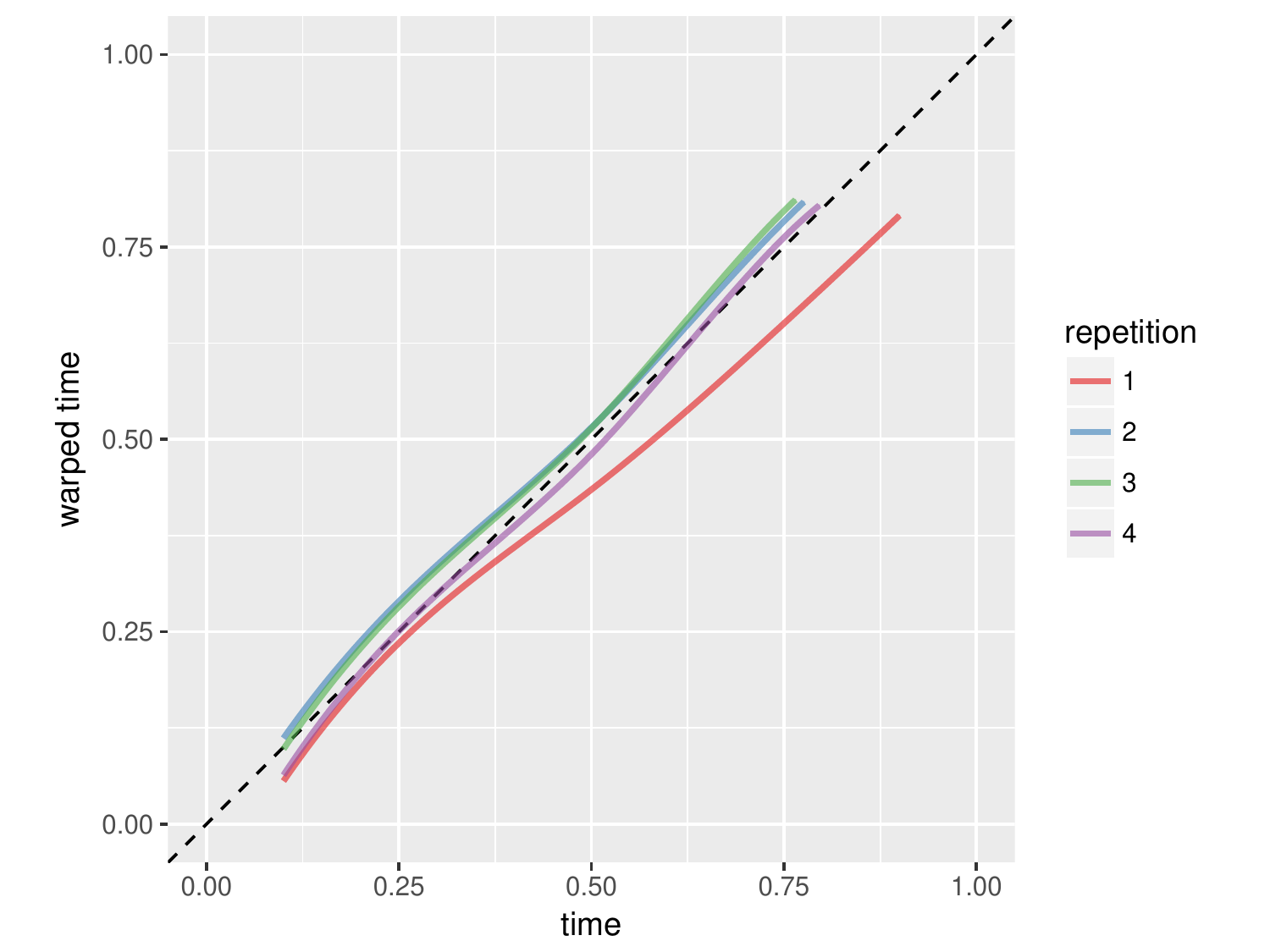}
	\caption{Predicted warping functions for the motion capture data}\label{fig:mcd-warp}
\end{figure}

\begin{figure}[!tp] 
	\centering
	\makebox[0.42\textwidth]{\textsf{observed}}\makebox[0.42\textwidth]{\textsf{aligned}}
	\includegraphics[width = 0.42\textwidth]{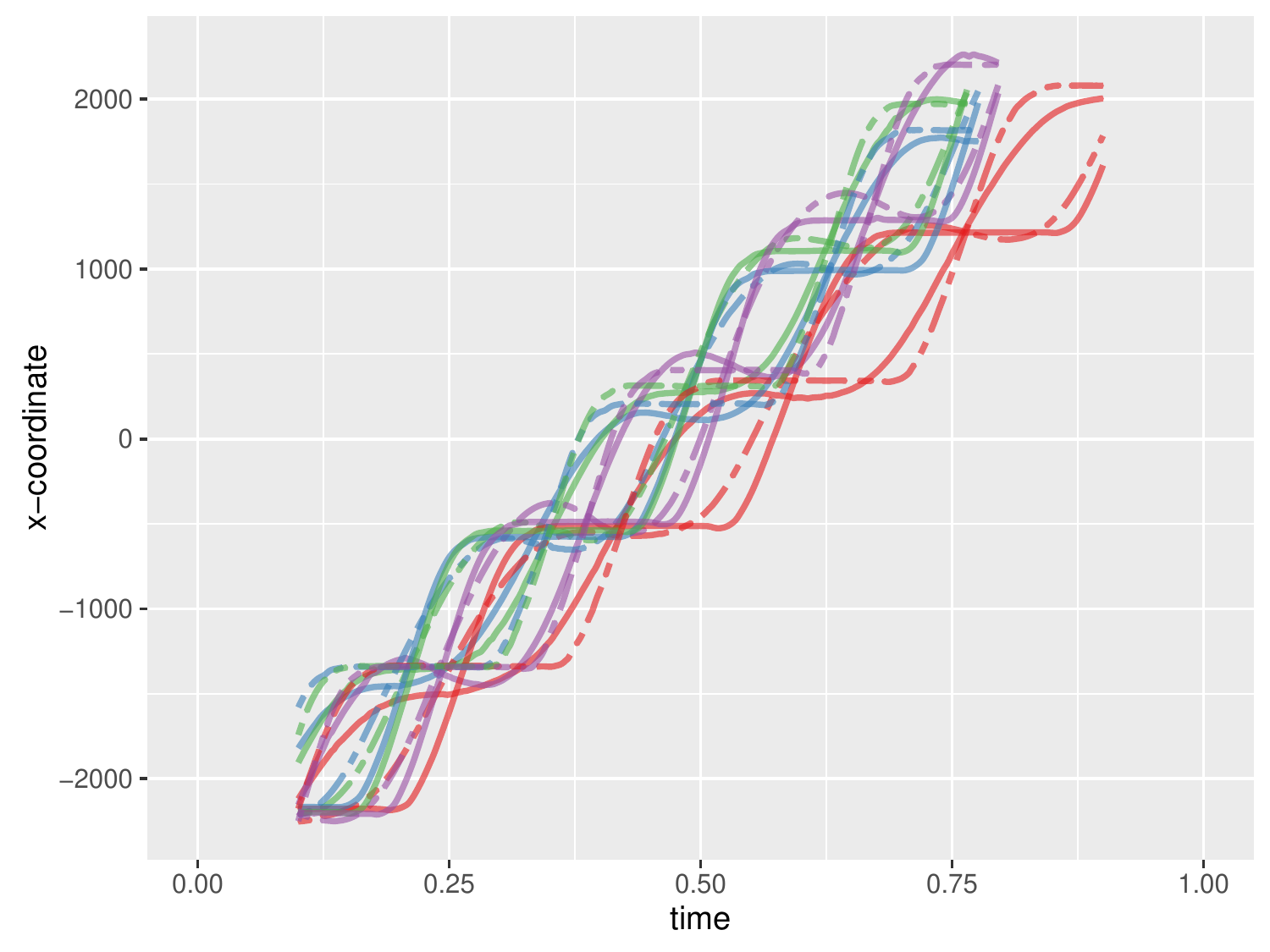}
	\includegraphics[width = 0.42\textwidth]{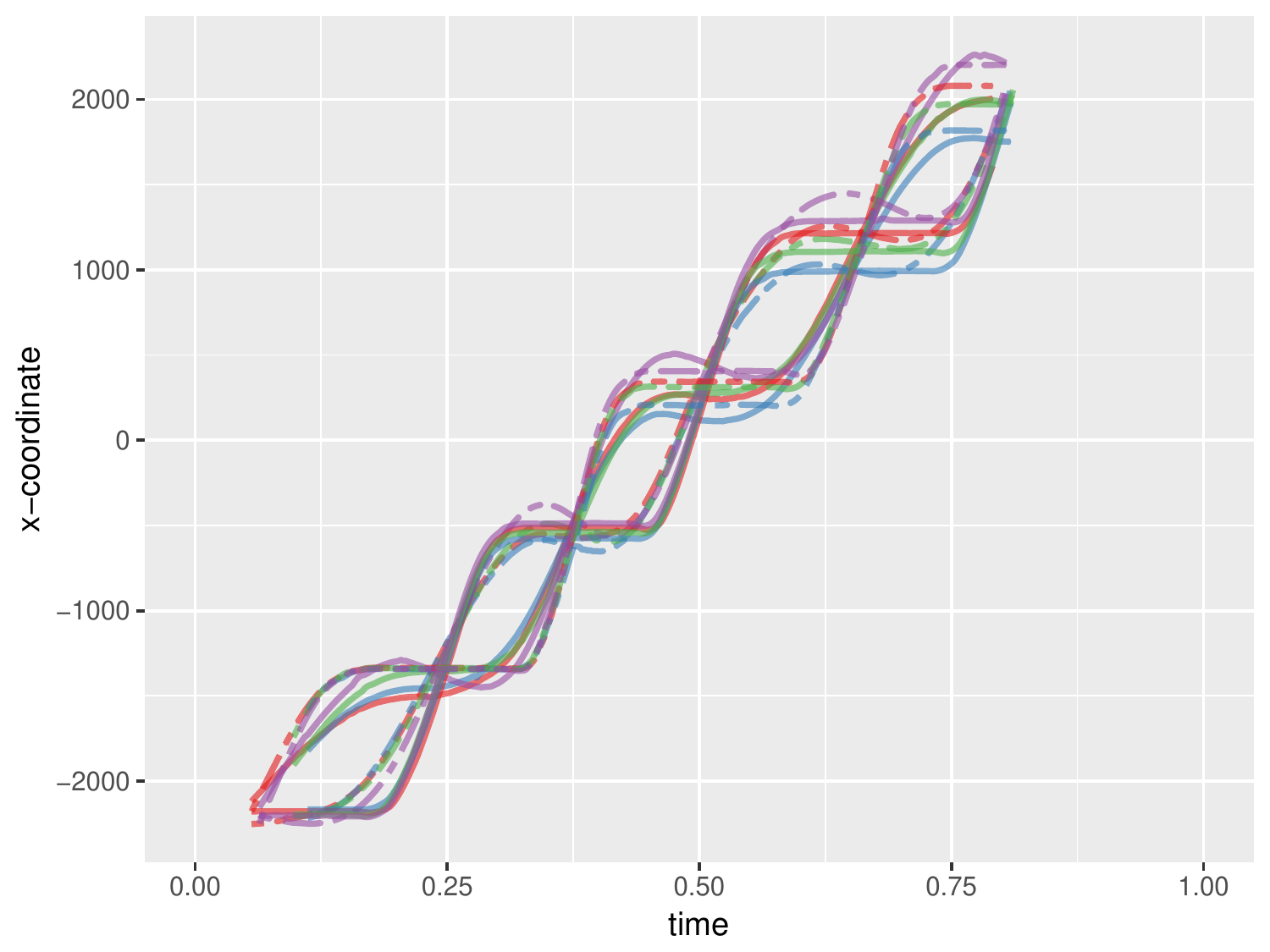}
	\includegraphics[width = 0.42\textwidth]{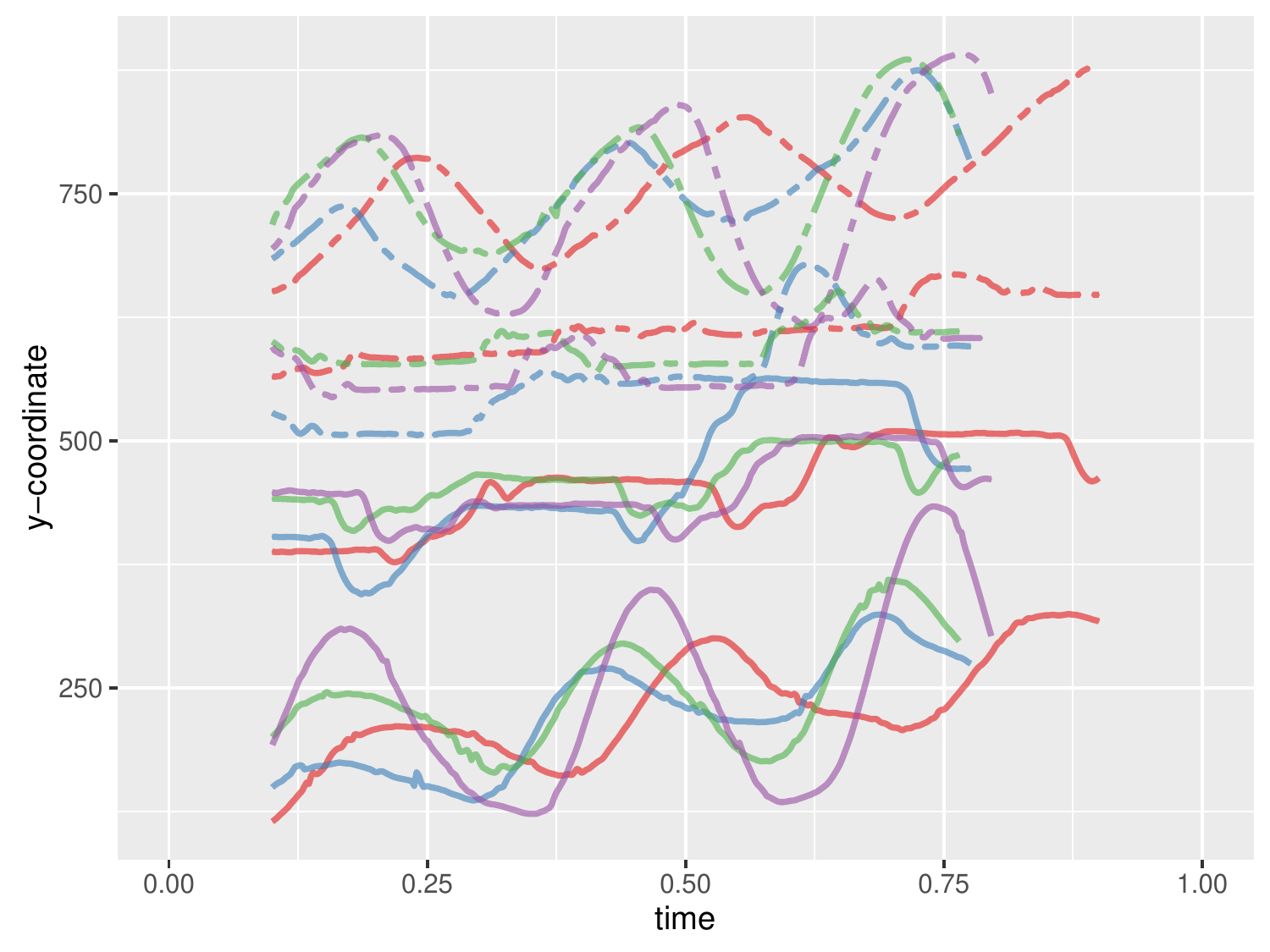}
	\includegraphics[width = 0.42\textwidth]{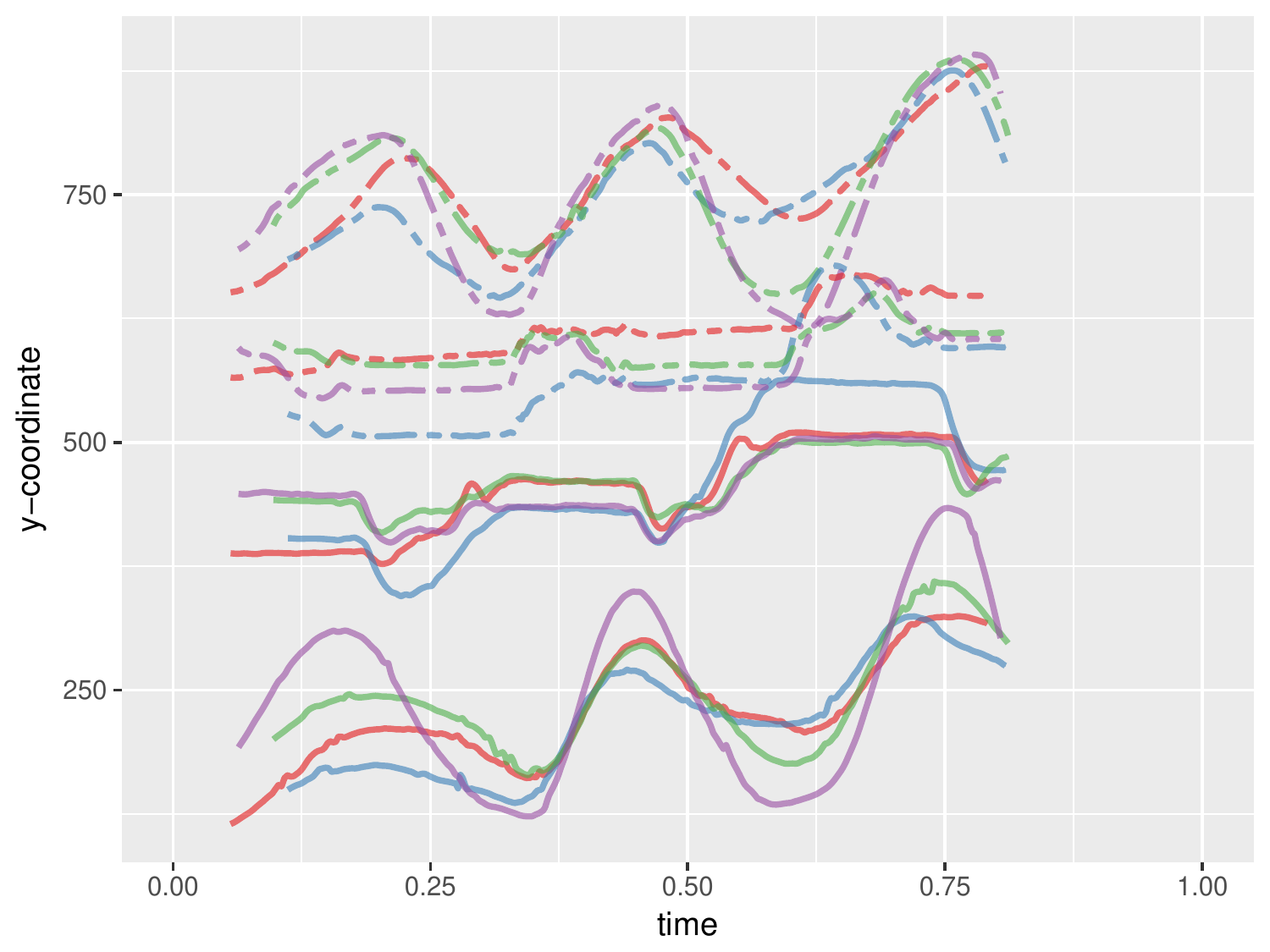}
	\includegraphics[width = 0.42\textwidth]{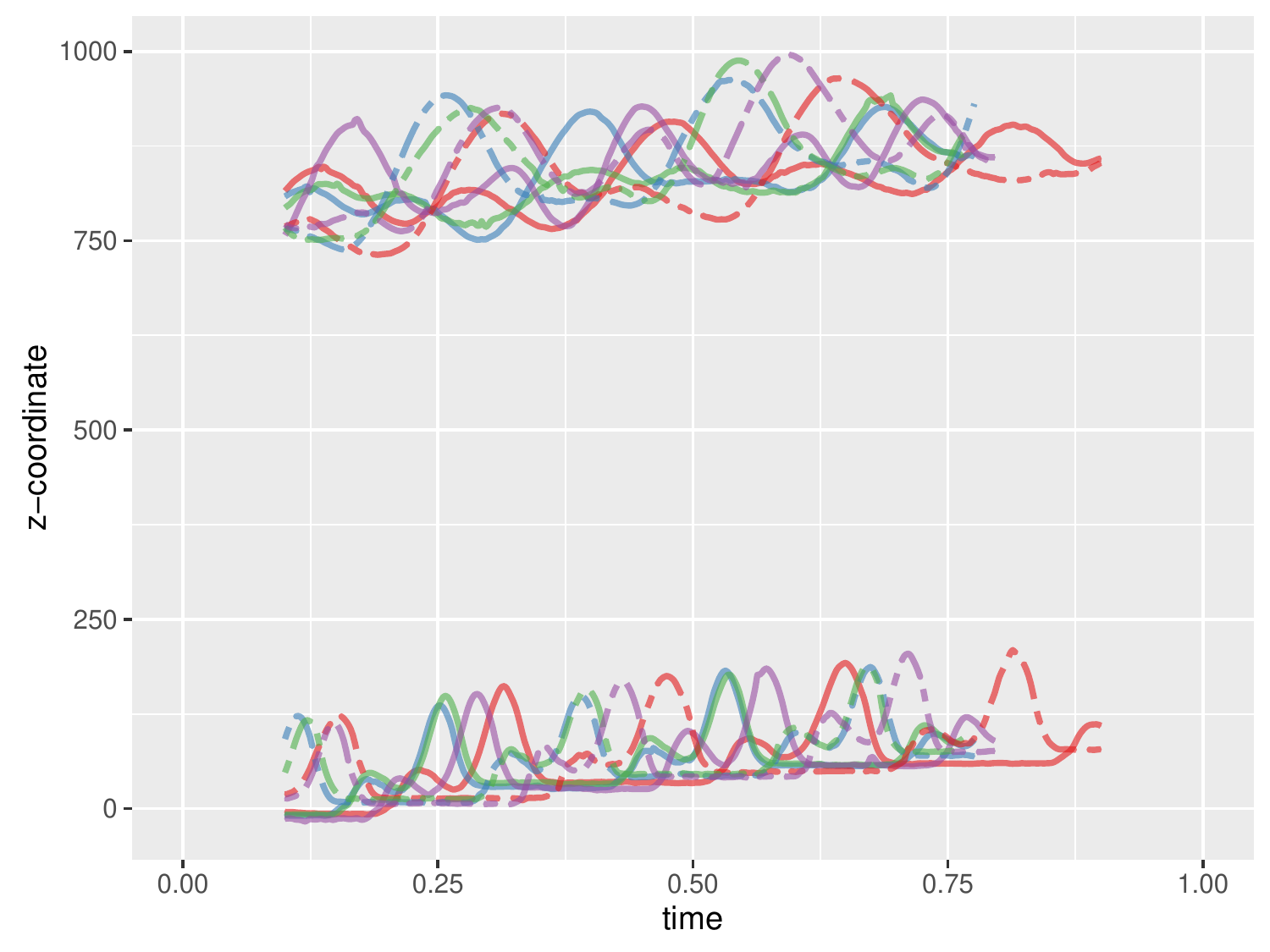}
	\includegraphics[width = 0.42\textwidth]{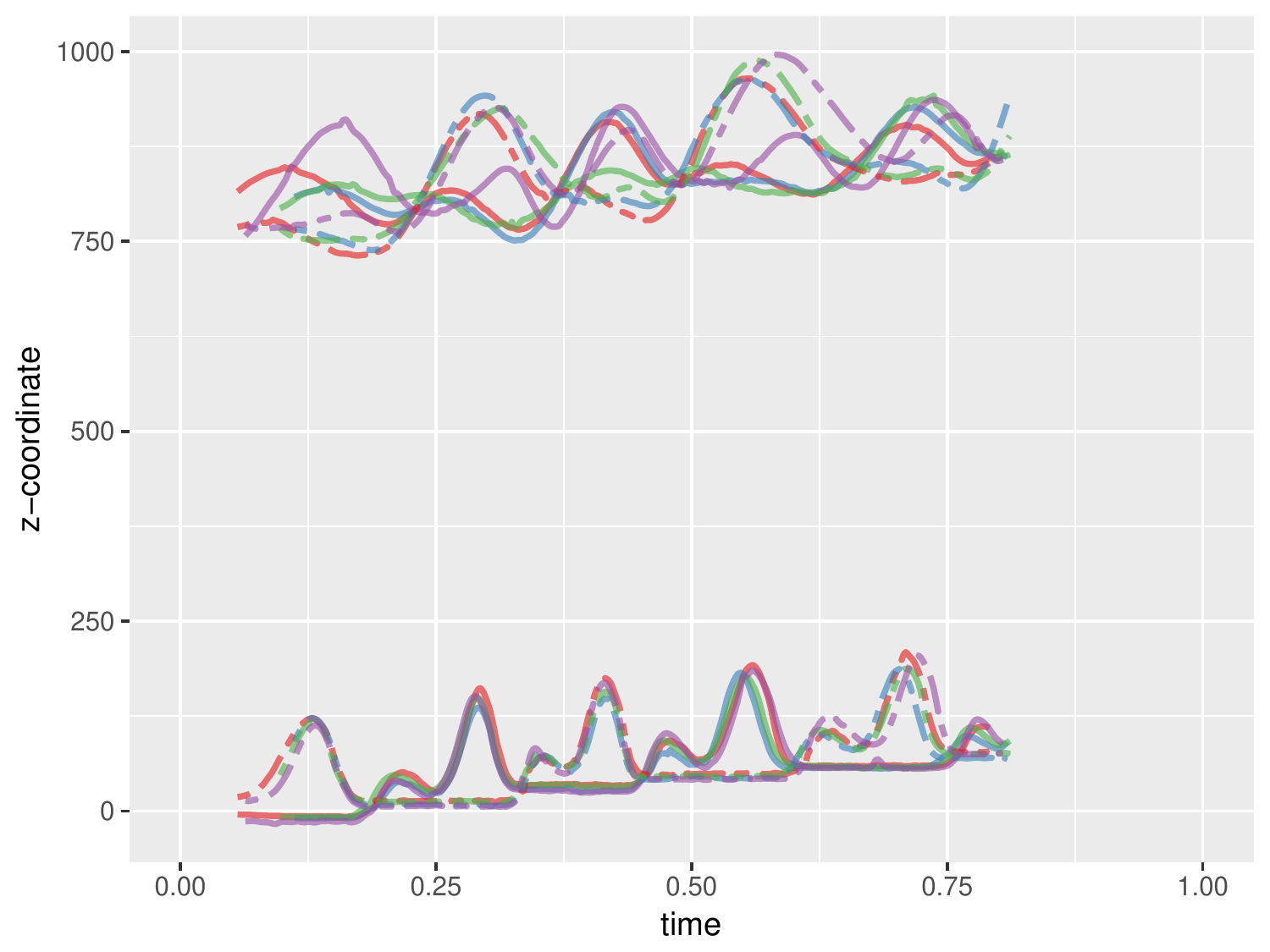}
	
	\includegraphics[scale = 0.5]{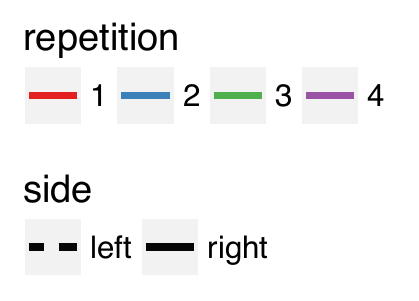}
	\caption{Observed and aligned curves from the motion capture data. Data values are the raw values from the tracking system.} \label{fig:mcd-kurver}
\end{figure}
\begin{figure}[!tp]
	\centering
	\includegraphics[width = \textwidth, trim = 0 30 0 0, clip]{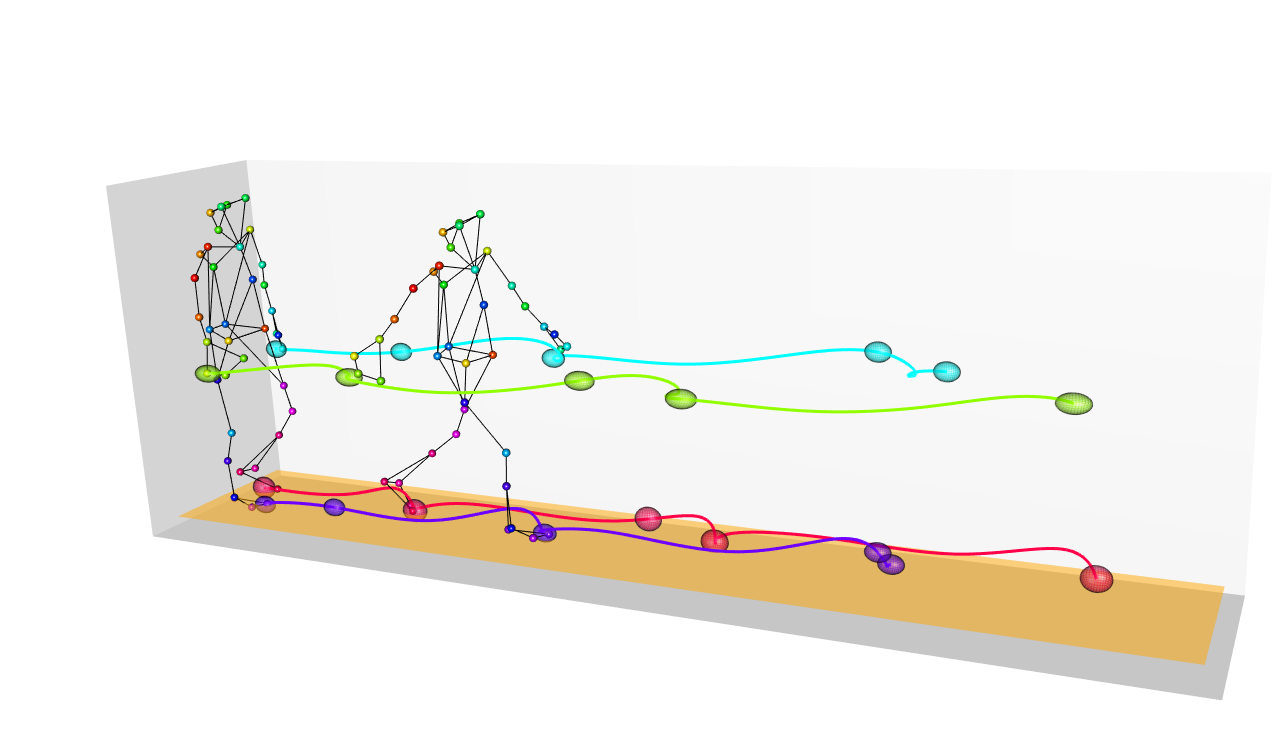}
	\caption{Estimated mean trajectories with five temporally equidistant ellipsoids indicating 95\% (marginal) confidence areas.  Two of the intermediate body poses of the fourth sample have been added as a reference.}  \label{fig:mcd-med-estimater} 
\end{figure}

\newpage
\subsection{Height and weight data}
Consider the height and weight measurements from the Copenhagen Puberty Study \citep{aksglaede2009recent, sorensen2010recent} shown in Figure~\ref{fig:hw_motivation}. The data contains 960 pairs of height and weight measurements for 106 healthy Danish boys.  The individual amplitude effects in the data set are clearly visible in the form of systematic deviations from the mean. The data also contain warping variation in the sense that age is a proxy for developmental age; each boy has his own internal clock that determines, for example, the onset of puberty. Alignment for this warping effect would then align the pubertal growth spurts visible as steep height increase in the individual boys occuring in the period 11 to 14 years. 

\paragraph{Modeling} While height is a naturally increasing function of age, weight is not necessarily. However, looking at the 2014 Danish weight reference \cite{tinggaard20142014}, we see a convex increase in the cross-sectional mean weight curve in the relevant age interval. Based on this we modeled $\btheta$ using an increasing spline (integrated quadratic B-splines) basis with 20 equidistant internal knots in the age interval $[5, 17]$ in both dimensions. 
The warping functions~\eqref{eq:warpfct} were modeled as increasing cubic (Hyman filtered) splines with $m_{\bw}=3$ equidistant internal anchor points in the age interval $[5, 20]$ and extrapolation at the right boundary point as in Figure~\ref{fig:warp_interpolation}(b). The latent variables $\bw_n$ were modeled as discretely observed Brownian motions with a single scale parameter. The temporally increasing variance of the Brownian motion seems as a good model for developmental age where one would expect high initial synchronization, and up to several years desynchronization at the onset of puberty.  

To model the amplitude variation, we used a dynamic cross-covariance with equidistant knots at $\{5, 10, 15, 20\}$ years as described in Proposition~\ref{prop-dyn-cov}, that is,
\begin{equation*}
\mathcal{S}(s,t) = f_\text{Mat\'{e}rn(2,$\kappa$)}(s,t) B_s^\top B_t.
\end{equation*}
The temporal covariance structure $f_\text{Mat\'{e}rn(2,$\kappa$)}(s,t)$ is the Mat\'ern covariance function with fixed smoothness parameter $\alpha = 2$ and unknown range parameter $\kappa$, see equation~\eqref{cov:matern} in the supplement. This implies twice differentiable sample paths of $\bx_i$, which is a reasonable assumption given the nature of the data. Furthermore, since we expected heterogeneous variances of the measurement error $\bepsilon_{nk}$ on height and weight in equation~\eqref{eq:mod1b}, we extended the model with a parameter $\rho > 0$ such that 
\begin{equation*}
\var(\bepsilon_{nk}) = \sigma^2 \vektor{1 & 0 \\ 0 & \rho}.
\end{equation*}
This gives a total of 14 parameters describing the cross-covariance model.

\paragraph{Results}

\begin{figure}[!tp] 
	\centering
	\makebox[0.45\textwidth]{\textsf{observed}}\makebox[0.45\textwidth]{\textsf{aligned}}
	\includegraphics[width = 0.45\textwidth]{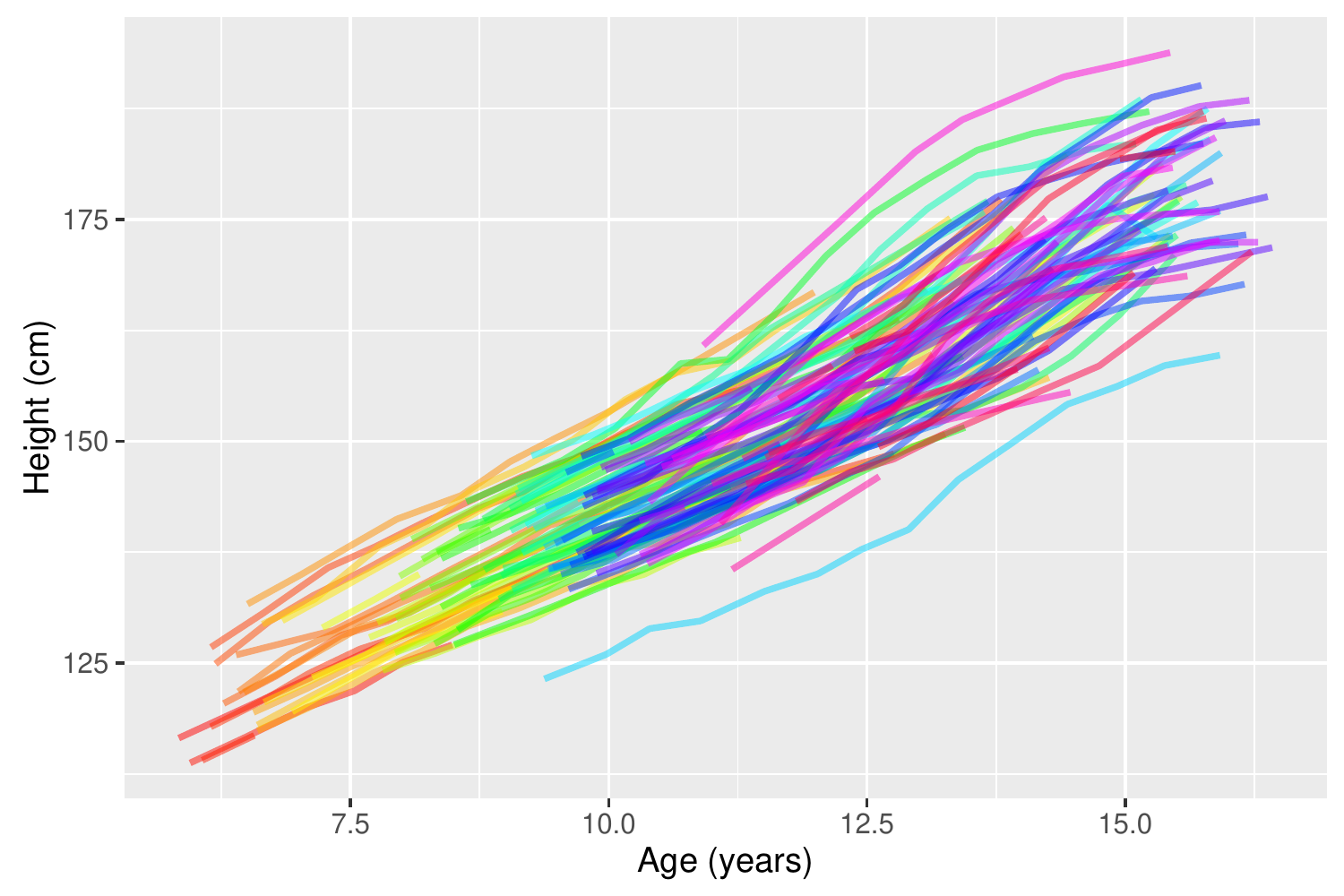}
	\includegraphics[width = 0.45\textwidth]{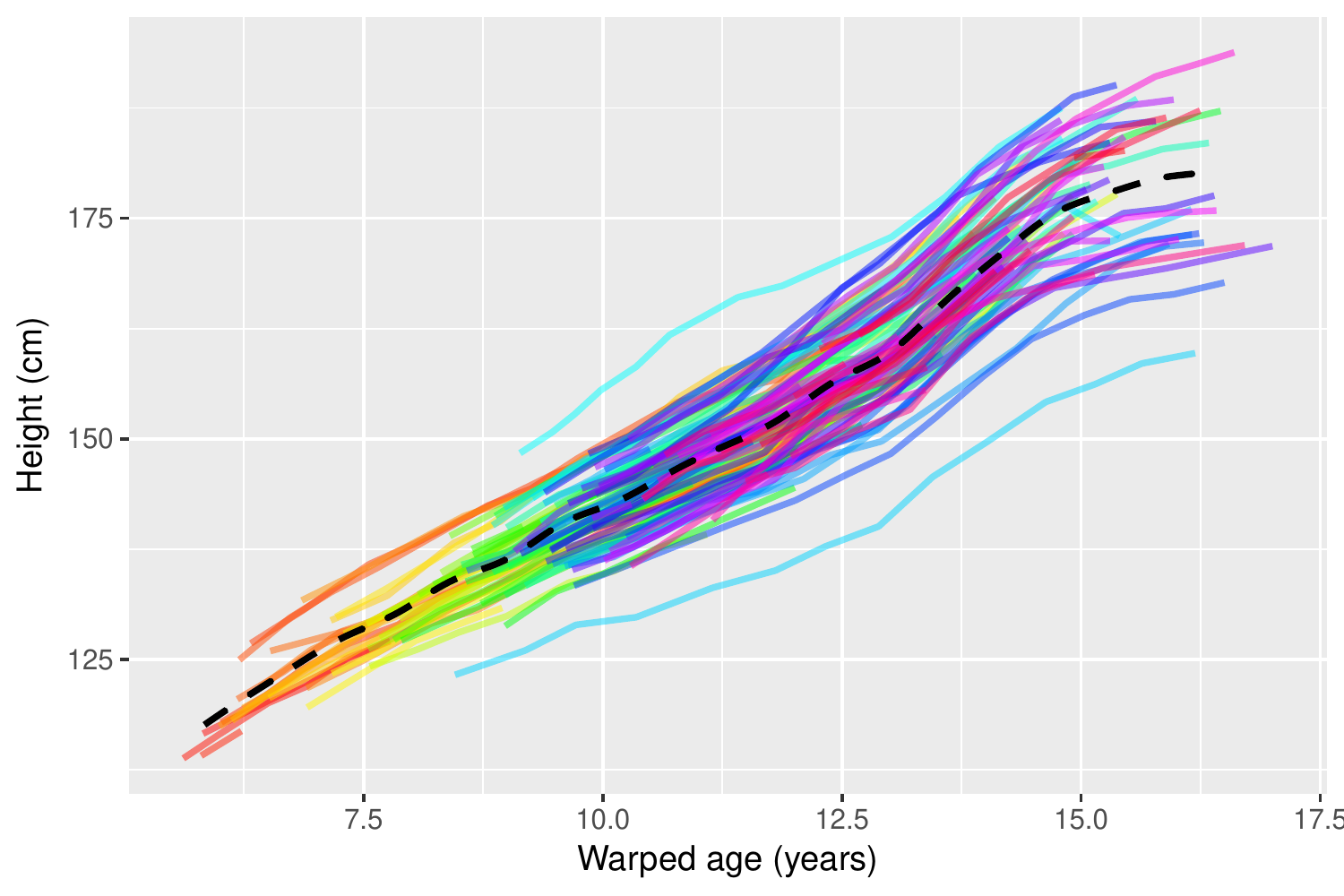}
	\includegraphics[width = 0.45\textwidth]{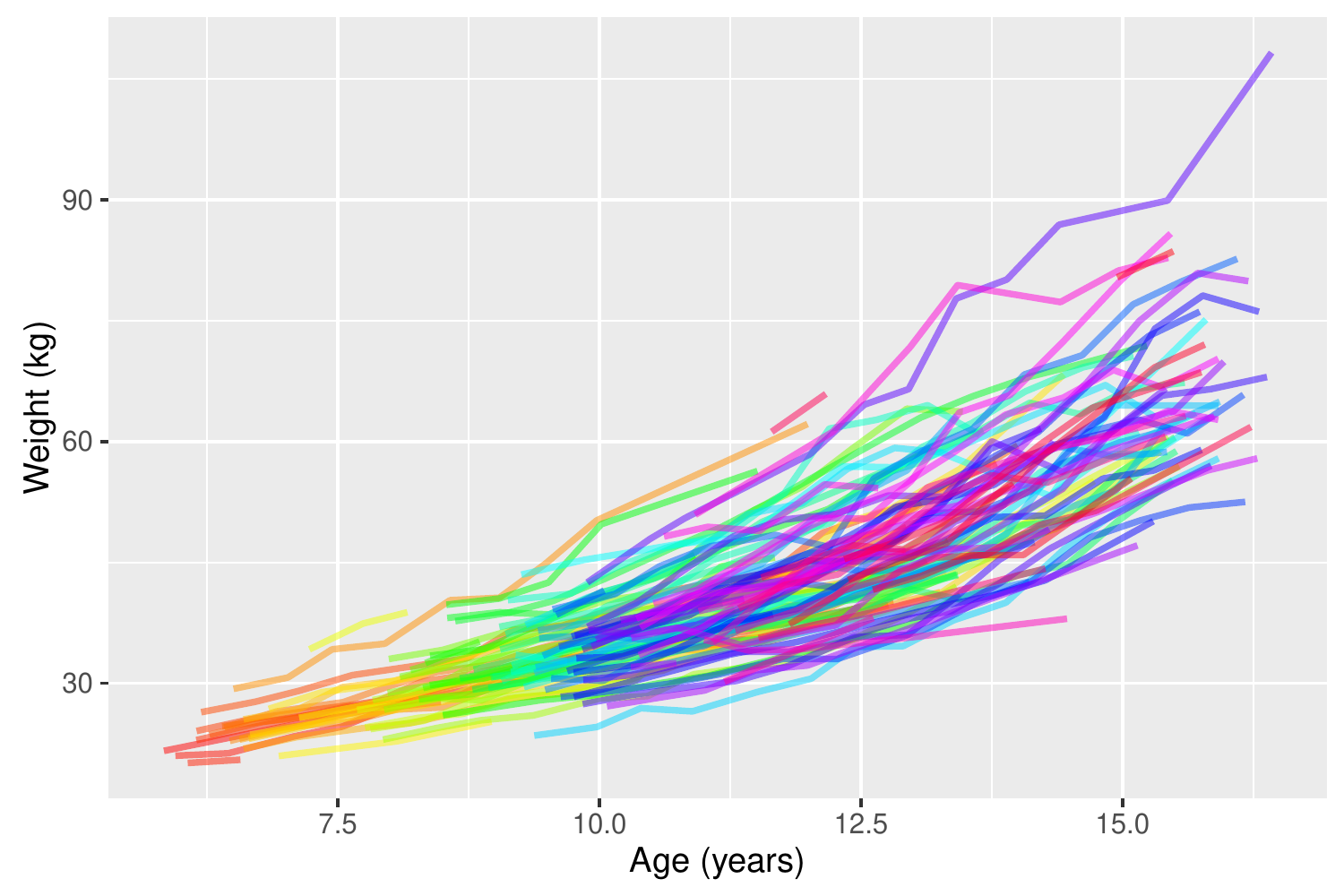}
	\includegraphics[width = 0.45\textwidth]{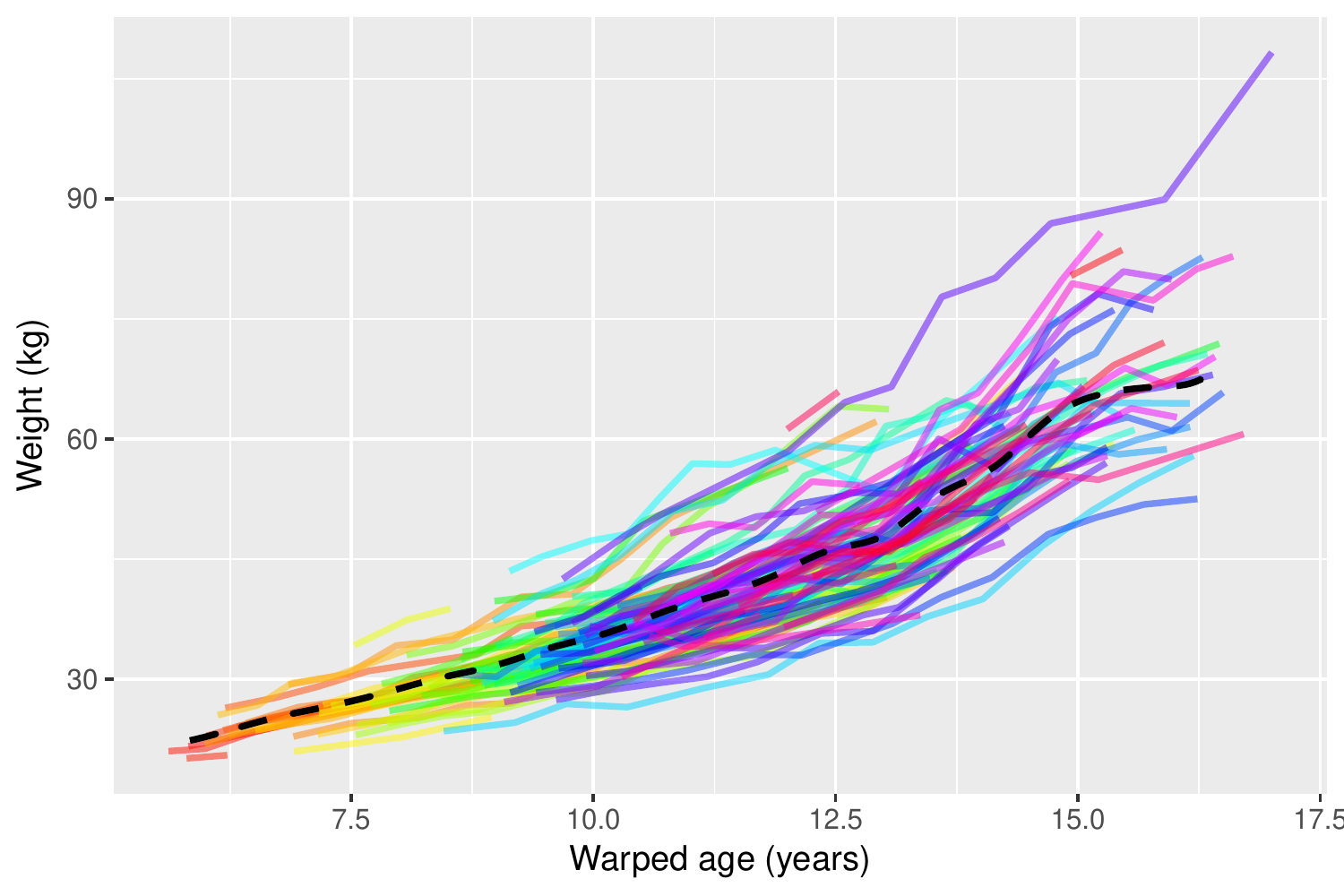}
	\caption{Observed and aligned height and weight curves from the Copenhagen Puberty Study. The estimated template curves are displayed as dashed black lines.} \label{fig:height_weight_results}
\end{figure}

\begin{figure}[!p] 
	\centering
	\includegraphics[width = 0.43\textwidth]{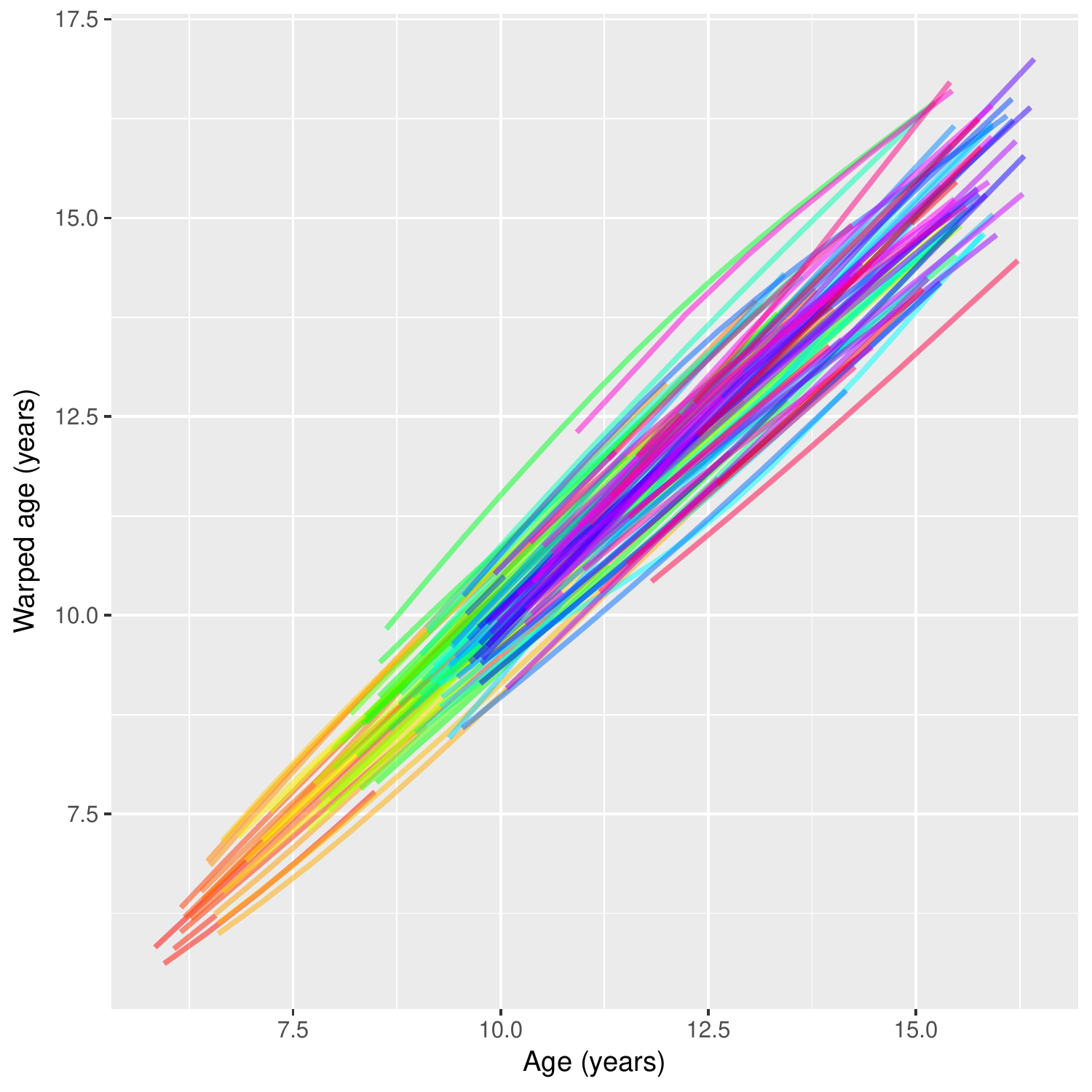}
	\caption{Predicted warping function corresponding to the data in Figure~\ref{fig:height_weight_results}.}\label{fig:height_weight_warps}
\end{figure}

\begin{figure}[!p] 
	\centering
	\includegraphics[width = 0.43\textwidth]{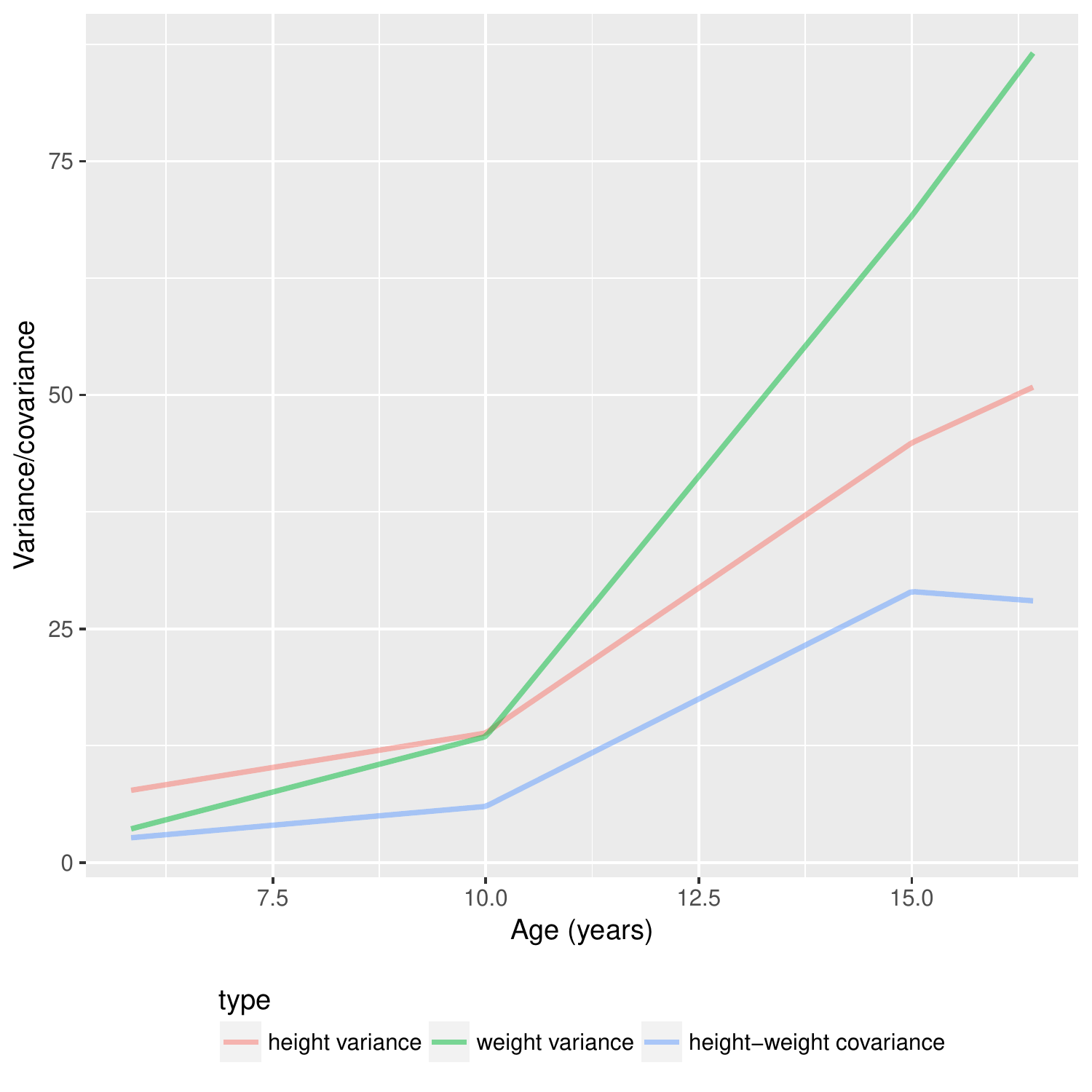}
	\caption{Estimated marginal variances and cross-covariance functions of age for the height and weight data in Figure~\ref{fig:height_weight_results}. The marginal variances also include the error variance. }\label{fig:height_weight_covariance}
\end{figure}

The aligned samples and estimated means are displayed in the right-side panels of Figure~\ref{fig:height_weight_results}, and the corresponding predicted warping functions can be found in Figure~\ref{fig:height_weight_warps}. We see that the individual growth curves are now aligned more tightly than before, in particular the pubertal height spurts seem to be well aligned. Although the shapes of the curves are well aligned, the model still allowed for considerable amplitude variation to be left after warping. This is as it should be; for increasing curves such as these a perfect fit could be achieved by warping, but the result would be meaningless and indicate that developmental age could be perfectly determined from a single measurement of a child's height. Given the proposed model-based separation of amplitude and warping effects induced by the maximum likelihood estimates, the information contained in a child's longitudinal data about the child's developmental age can be quantified through the posterior distribution of the warping effects. 

The estimated covariance structure is shown in Figure~\ref{fig:height_weight_covariance}. As one would expect, height and weight variances increase with age. The covariance increases at a slower rate and has a slight decrease after 15 years, giving a correlation of 0.42 at 16.5 years.

\subsection{Arm movement data}
\begin{figure}[!p]
	\centering
	\includegraphics[width = 0.44\textwidth, trim = 80 20 120 140, clip]{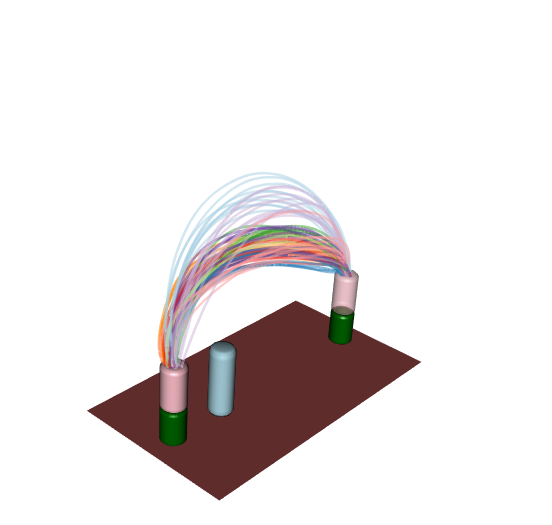}
	\includegraphics[width = 0.44\textwidth, trim = 80 20 120 140, clip]{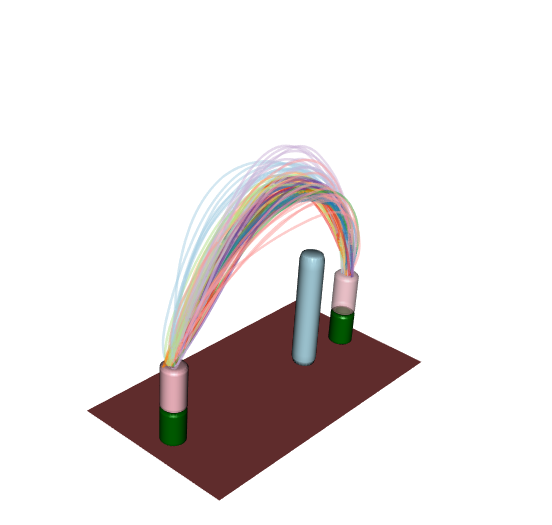}
	
	\includegraphics[width = 0.44\textwidth]{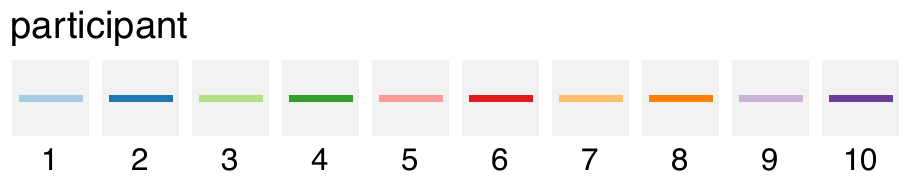}
	\caption{Recorded movement paths in experiment with small obstacle 15 cm from starting position (left) and tall obstacle 45 cm from starting position (right).}\label{fig:movement1}
\end{figure}

Our third example is an analysis of human arm movements in obstacle avoidance tasks. Hand-movement paths in two experimental conditions are displayed in Figure~\ref{fig:movement1}.  In each experimental condition, a wooden cylindrical object (pink) located at a starting position (green cylinder) was to be moved 60 centimeters forward and placed on a target cylinder. Between the starting and target positions, a cylindrical obstacle was placed. The obstacle height (small, medium, tall) and obstacle position (five equidistant positions between starting and target positions) varied with experimental condition. A total of 15 obstacle avoidance conditions were performed plus a control condition with no obstacle. Ten right-handed participants performed ten repetitions of each experimental condition, and the spatial position of the hand was recorded at a sampling rate of 110 Hz. The data set thus consists of 1600 functional samples with a total of $m = 175,535$ three-dimensional sampling points giving a total sample size of $526,605$ observations. The present data set is described in detail in \cite{Grimme2014}, and the experiment is a refined version of the experiment described in \cite{grimme2012naturalistic}. The data set is available through a public repository.\footnote{\url{https://github.com/larslau/Bochum_movement_data}}

\paragraph{Data processing and modeling} We analyzed the data separately for the 16 experimental conditions. Following the convention for modeling human motor control data, time was modeled as percentual time rather than observed time. This means that all movement time intervals were scaled to $[0,1]$, such that 0 corresponds to the onset of the movement and 1 corresponds to the end of the movement. We used model~\eqref{eq:mod1} to model the data separately for the 16 different experimental conditions. The mean path $\btheta_j$ for the $j$th participants was modeled in a cubic B-spline basis with 21 interior knots. We modeled the warping functions~\eqref{eq:warpfct} as increasing cubic spline interpolations (Hyman filtered) with $m_{\bw}=3$ equidistant anchor points. The choice of three knots was evaluated, and found optimal, in terms of the cross-validation set-up described in the classification study below. The latent variables $\bw_n$ were modeled as discretely observed Brownian bridges with a single scale parameter, because of the fixed endpoints of the data. 

The amplitude variation was modeled using a dynamic cross-correlation model with knots at $\{0, 0.4, 0.6, 1\}$ as described in Proposition~\ref{prop-dyn-cov}, that is,
\begin{equation*} 
\mathcal{S}(s,t) = f_\text{mixture($a$)}(s,t) f_\text{Mat\'{e}rn($\alpha$,$\kappa$)}(s,t) B_s^\top B_t.
\end{equation*}
The temporal covariance structure is given as a combination of stationary and bridge Mat\'ern serial correlation with mixture parameter $a$, smoothness parameter $\alpha$, and range parameter $\kappa$. The details of this covariance structure are described in equations~\eqref{cov:mixture} and \eqref{cov:matern} in the supplement. This dynamic cross-correlation structure has $27$ free parameters.

The knot positions $\{0, 0.4, 0.6, 1\}$ were chosen such that we were able to model a  change in cross-correlation structure around the middle of the movement in percentual time, in particular the change that happens when the movement progresses from lift to descend. The concept of isochrony \citep{grimme2012naturalistic} suggests that the times where the peak heights are reached are largely invariant to obstacle height and placement, and for the given data the peak heights generally occur for $t\in(0.4, 0.6)$, see for example Figure~\ref{fig:unaligned_aligned}.

The left column of Figure~\ref{fig:unaligned_aligned} displays the observed $x$-, $y$- and $z$-coordinates in a single experimental condition as functions of percentual time. The right column displays the coordinates in predicted warped percentual time.  We see that the $x$- and $z$-coordinates are very well aligned within participant, and that the alignment of the $y$-coordinate seems to contain a relatively larger proportion of amplitude variation after alignment than the $x$- and $z$-coordinates. We note that the alignment procedure does not change the movement path in $(x, y, z)$-space. The predicted maximum-a-posteriori warping functions are displayed in Figure~\ref{fig:pred_warp}.

\begin{figure}[!tp]
	\centering
	\includegraphics[width = 0.45\textwidth]{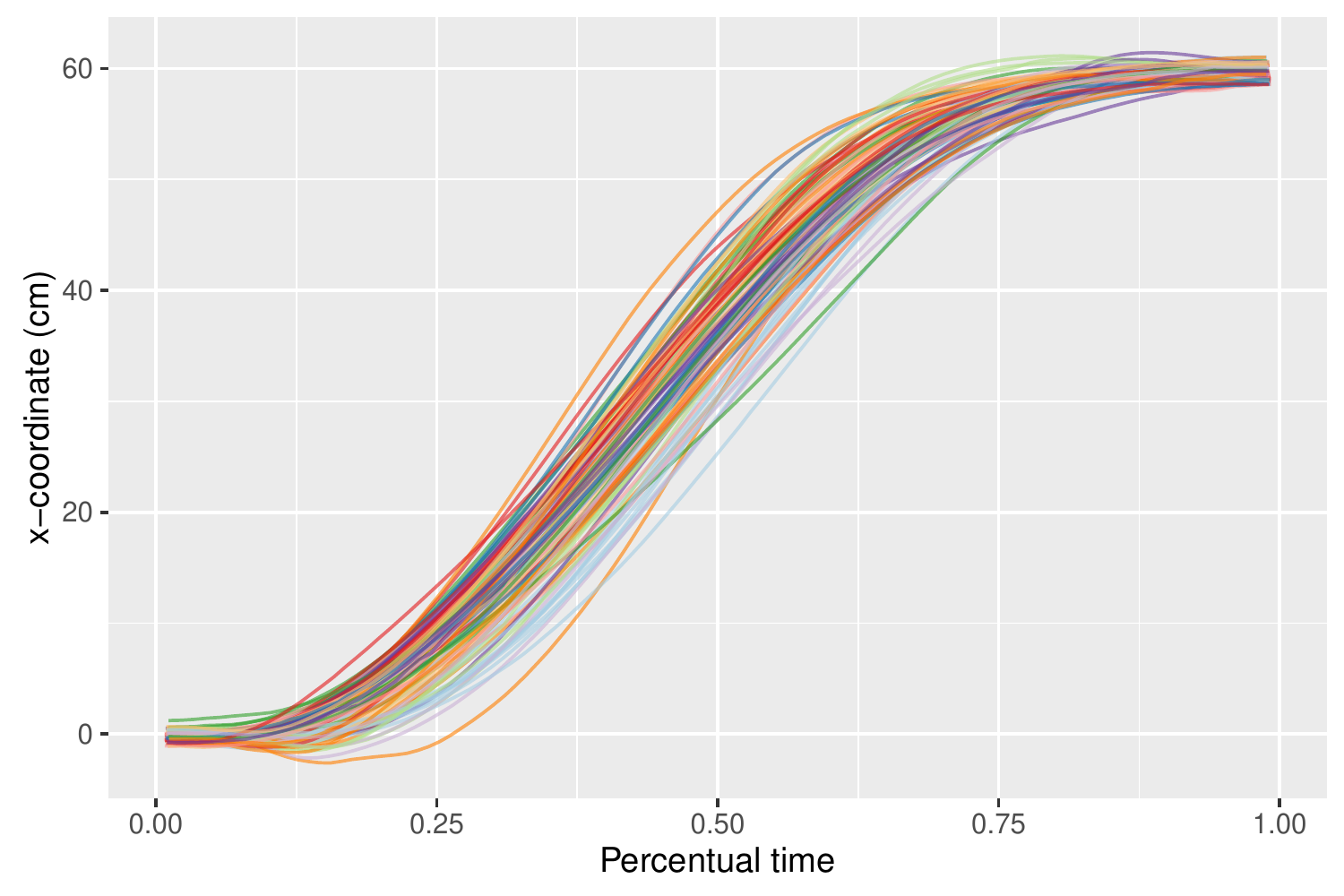}
	\includegraphics[width = 0.45\textwidth]{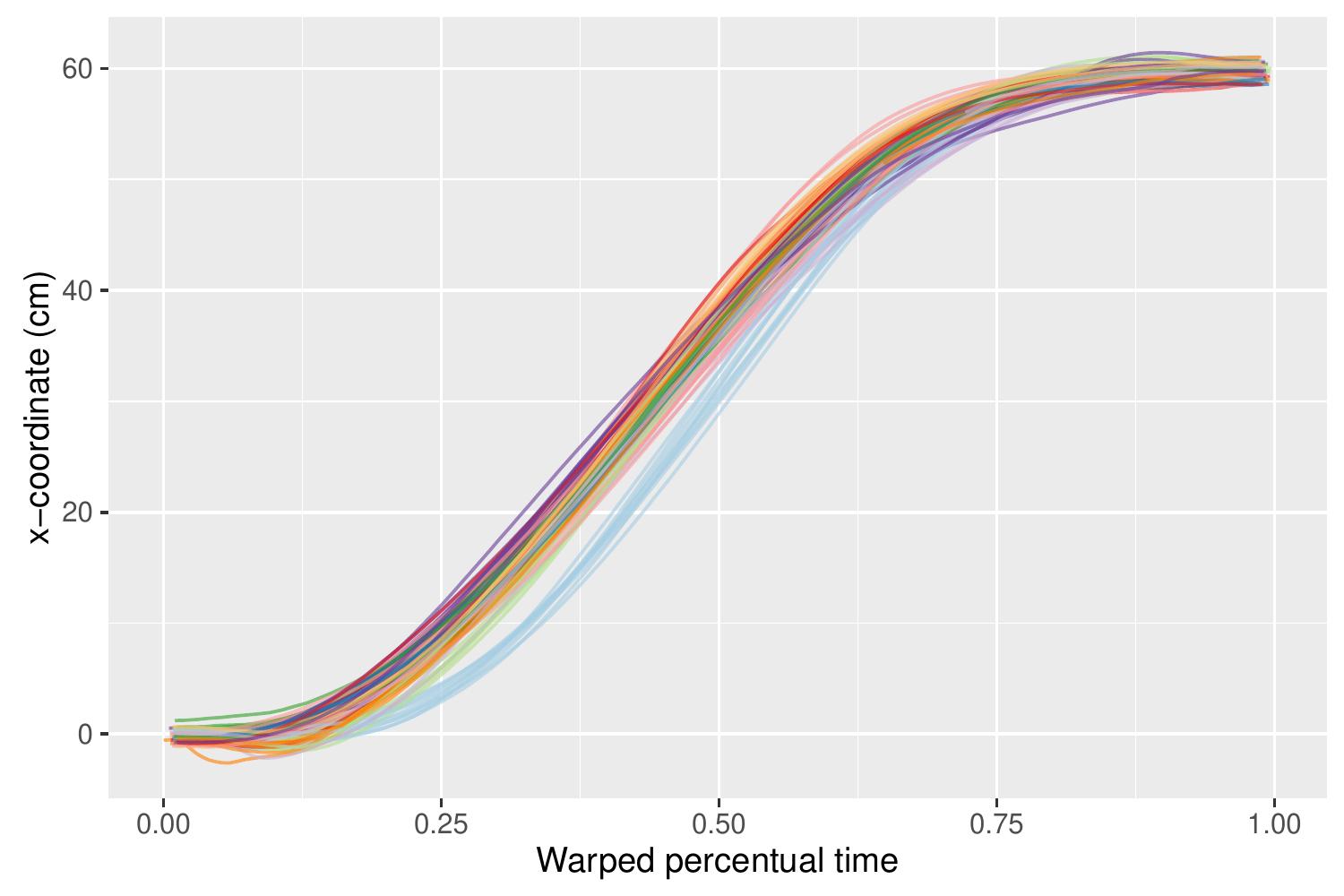}
	
	\includegraphics[width = 0.45\textwidth]{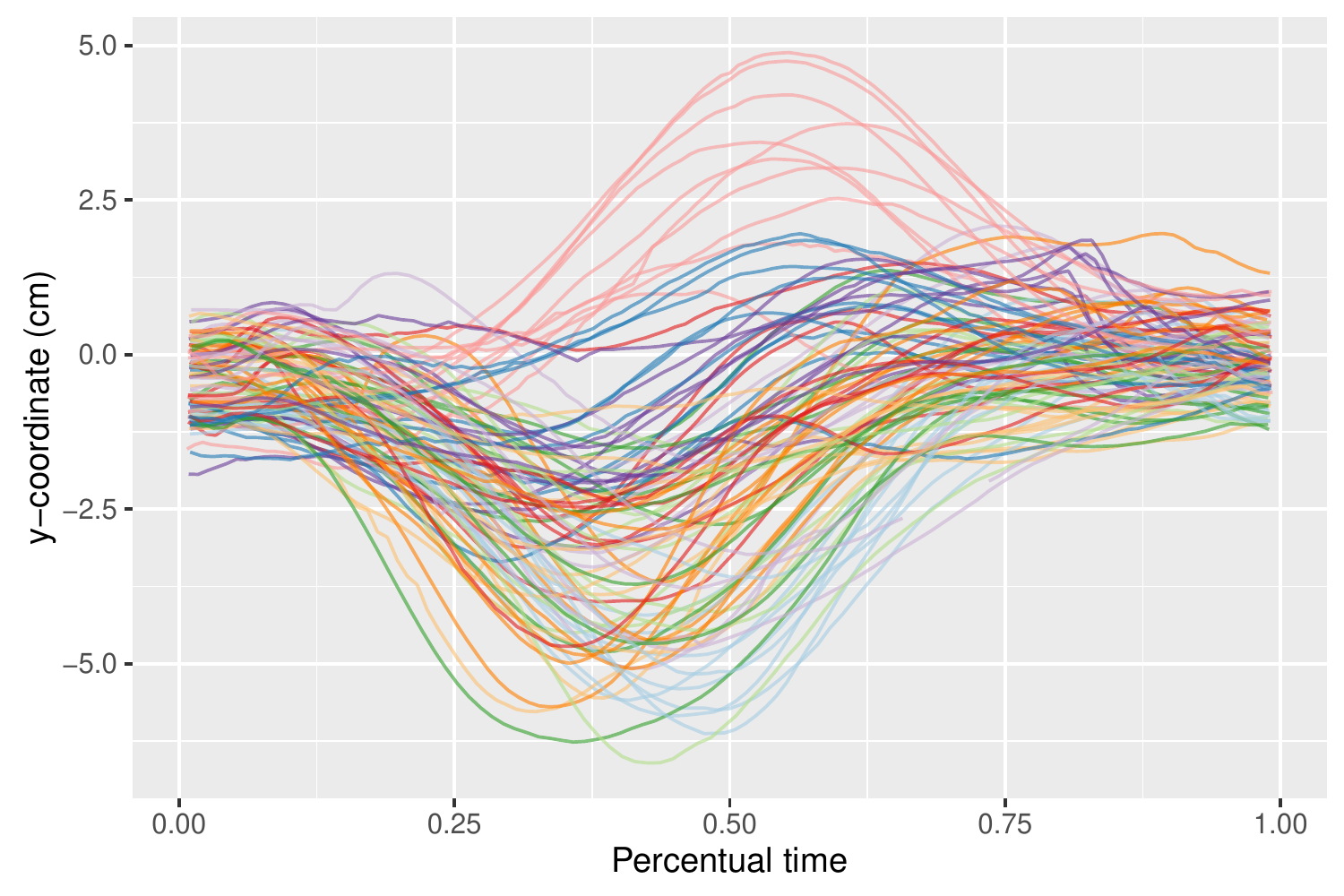}
	\includegraphics[width = 0.45\textwidth]{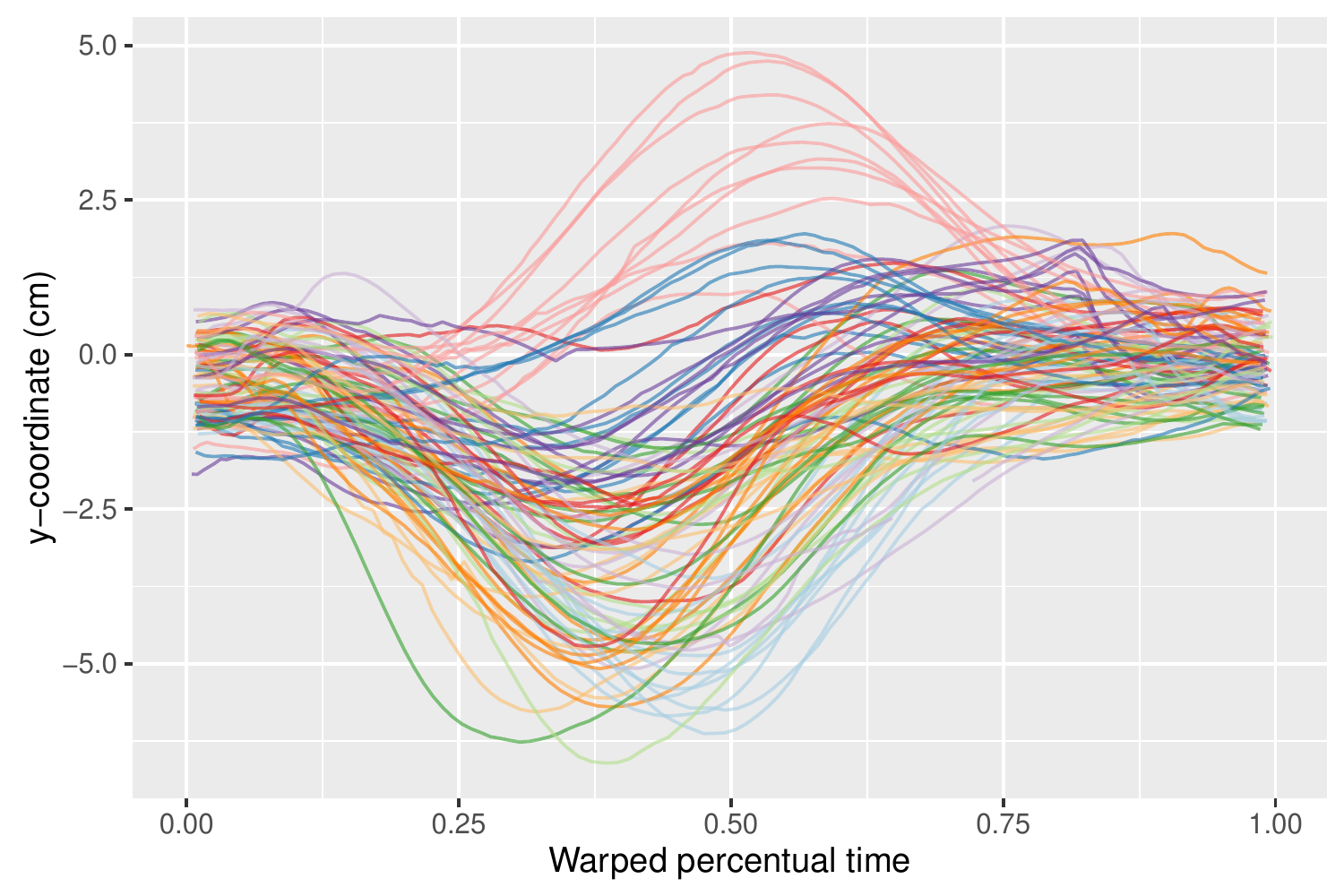}
	
	\includegraphics[width = 0.45\textwidth]{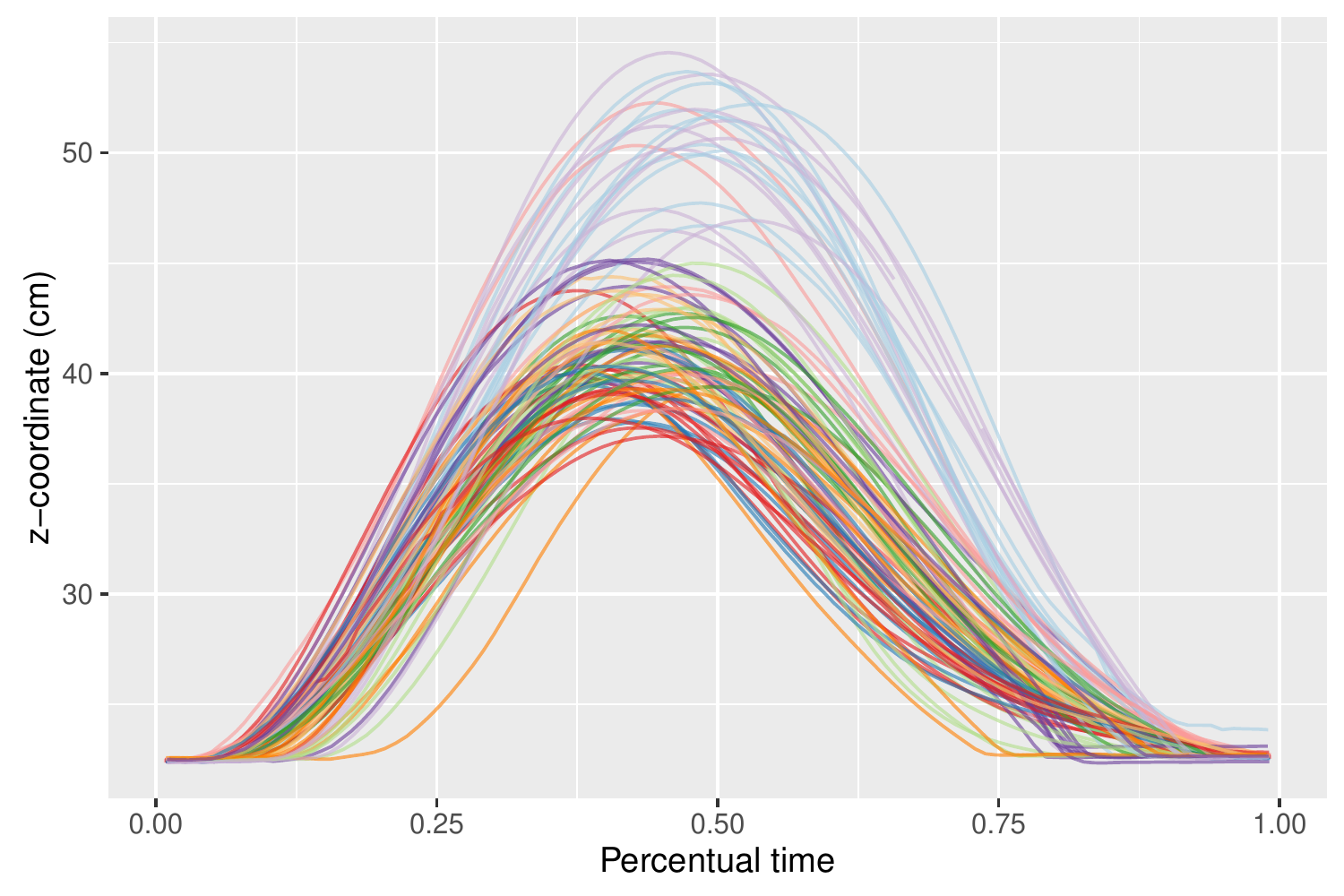}
	\includegraphics[width = 0.45\textwidth]{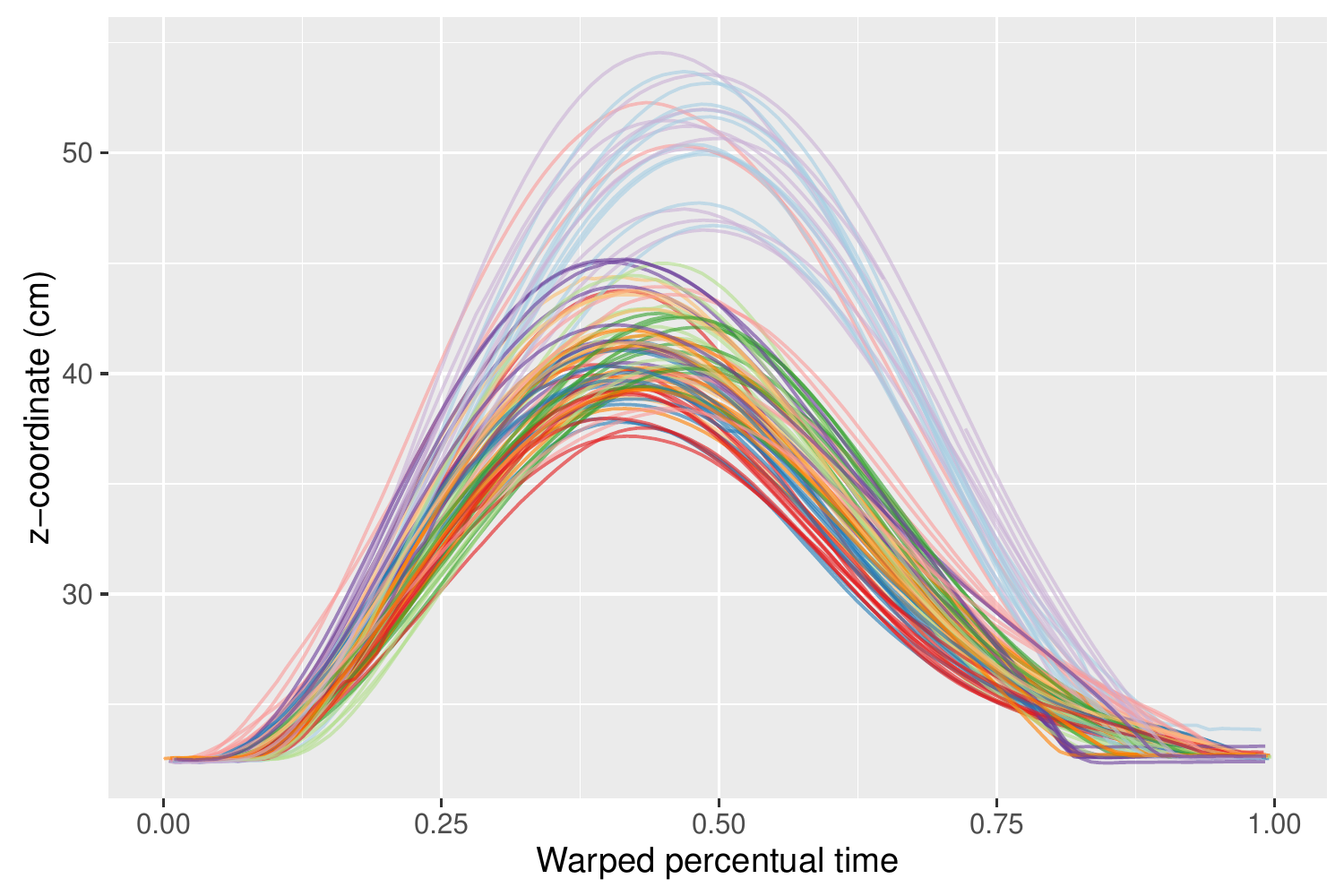}
	
	\caption{Data from the experiment with a small obstacle 30 cm from starting position plotted in percentual time (left column) and warped percentual time (right column). Coloring follows the coloring in Figure~\ref{fig:movement1}.}\label{fig:unaligned_aligned}
\end{figure}

\begin{figure}[!hpt]
	\centering
	\includegraphics[width = 0.45\textwidth]{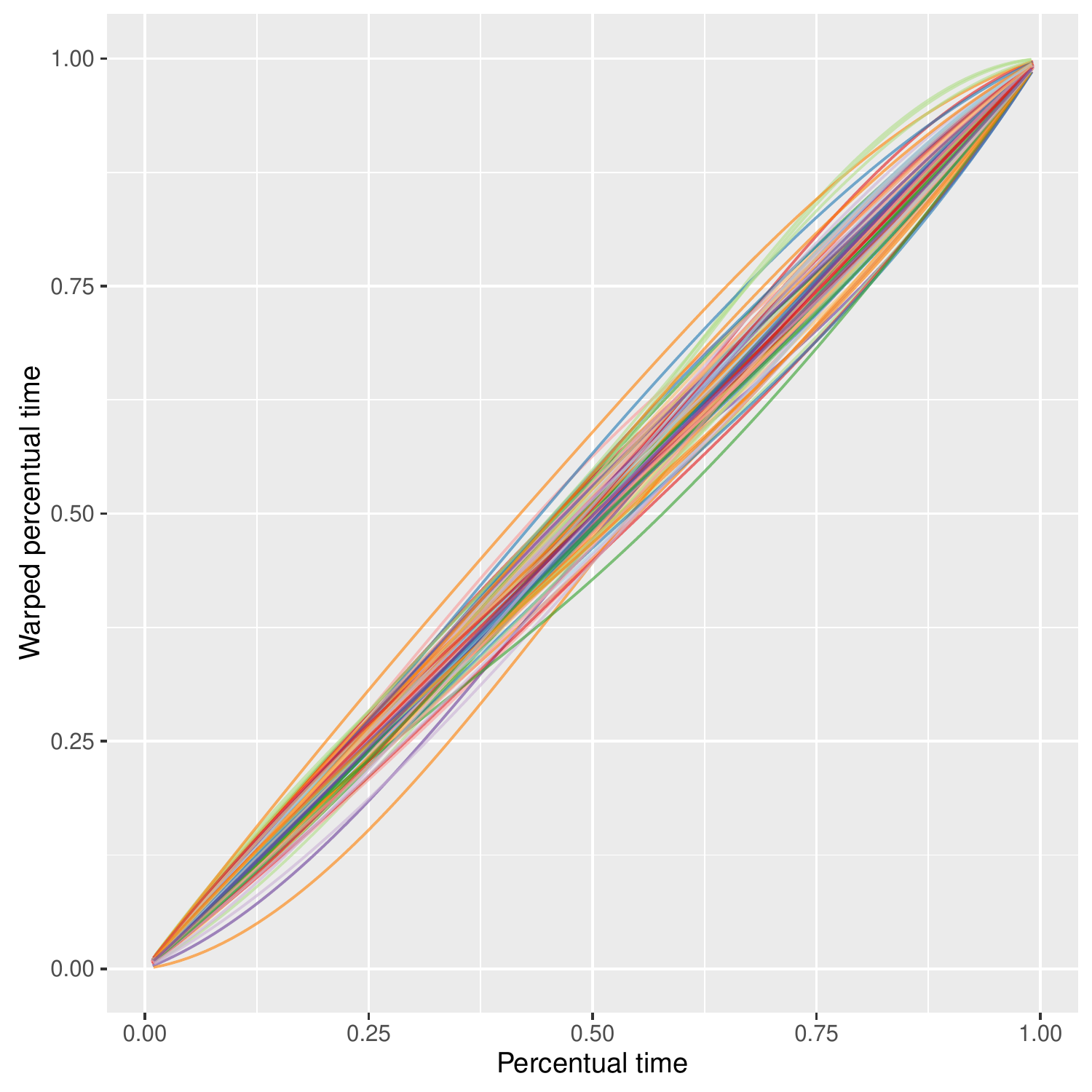}
	\caption{Predicted warping functions corresponding to the alignment in Figure~\ref{fig:unaligned_aligned}.}  
	\label{fig:pred_warp}
\end{figure}

\paragraph{Parameter estimates}
The common variance parameter $\sigma$ and the Mat\'ern parameters $\alpha$ and $\kappa$ varied little with experiment. On the other hand the relative weight, $a$, of the stationary covariance and the bridge covariance varied considerably across experiments. However $a$ was large in all cases meaning that a large majority of the variance is captured by the stationary part.  We refer to Table~\ref{table:suppl_param} in the supplementary material for all parameter estimates.

\paragraph{Variance and cross-correlations} 
The amplitude variation was assumed to be generated from Gaussian processes $\bx_n$ and white noise $\bepsilon_n \sim N(0, \sigma^2 \bI_{3m_n})$. Since the observed curves are very smooth the estimated contributions from the white noise terms were very small.  

Figure~\ref{move-V-S} show the ratios of systematic amplitude variance to linearized systematic variance (amplitude and linearized warp) as estimated by the model. At the endpoints all variance was captured by the serially correlated amplitude effect.
In the $y$-direction almost all variation was captured by the amplitude variance which fits well with the aligned $y$-coordinates of the movement path in Figure~\ref{fig:unaligned_aligned}. The warp-related variance accounted for a larger part of the variation in the $x$- and $z$-directions. The temporal structure of the $x$-coordinate reveals that the warp effect explained the majority of the variance around the middle of the movement, while for the $z$-coordinate it explained the majority of the variance during lift and descend. Thus, the model predicted warping functions using a trade-off where the (percentual) temporal midpoints of the transport component and the lift and descend components had highest influence when measuring the alignment of samples.

The individual participant's estimated mean trajectories and the systematic amplitude variation are illustrated in Figure~\ref{3d-cov1}. In the right-hand illustration, the prediction ellipsoids in the middle are relatively small considering that this is the region with most variation. This is because most of the variation was captured by the participant-specific mean curves and the warping effect, as one would expect. The amplitude variance around the endpoints seems somewhat overestimated, which suggests that the chosen anchor points provided a too coarse model for the dynamics of the true covariance function around the endpoints.

\begin{figure}[!tp]
	\centering
	\includegraphics[width=0.95\textwidth]{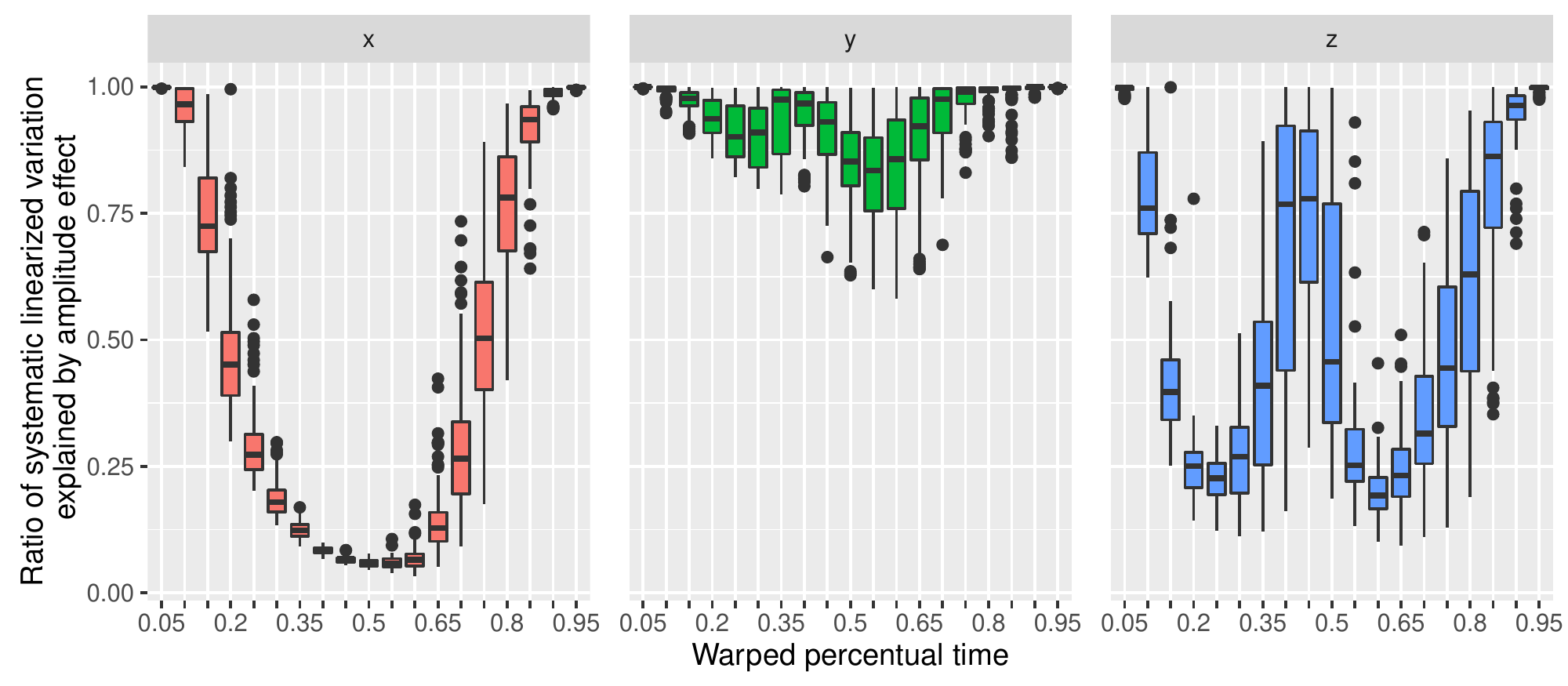}
	\caption{Coordinatewise boxplot of the temporal development of the ratio of $S_n$ to $S_n + Z_n C Z_n^\top$ for the 100 samples in the experiment with a small obstacle 30 cm from starting position.}
	\label{move-V-S} 
\end{figure}

\begin{figure}[!tp]
	\centering
	
	\includegraphics[width = 0.48\textwidth, trim = 0 100 0 50, clip = true]{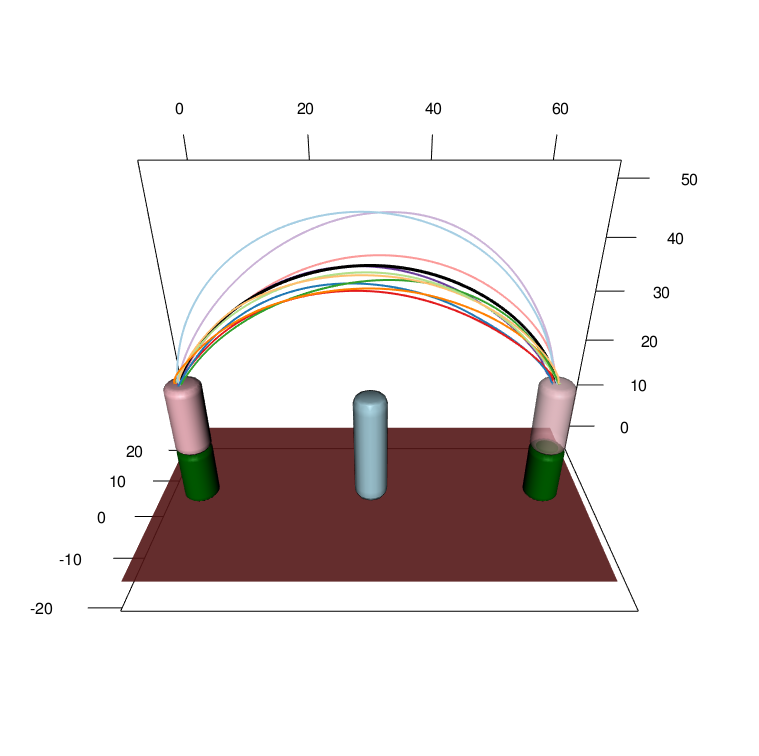}
	\includegraphics[width = 0.48\textwidth, trim = 0 100 0 50, clip = true]{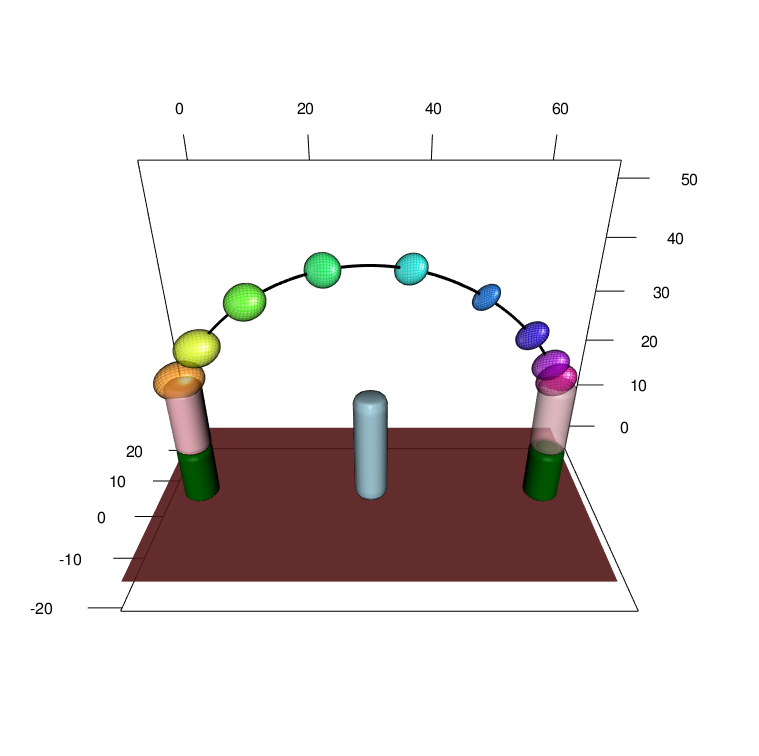}
	
	\includegraphics[width = 0.44\textwidth]{img/person.pdf}
	\caption{Estimated experiment-specific curve (black) and participant-specific curves for the experimental set-up with small obstacle 30 cm from starting position (left) and estimated 95\% predictions ellipsoids for the systematic amplitude effect in the same set-up (right). The ellipsoids are displayed temporally equidistant around the mean trajectory for the experimental set-up. 
	}\label{3d-cov1} 
\end{figure}

%
\begin{figure}[!tp]
	\centering
	\includegraphics[width=0.85\textwidth]{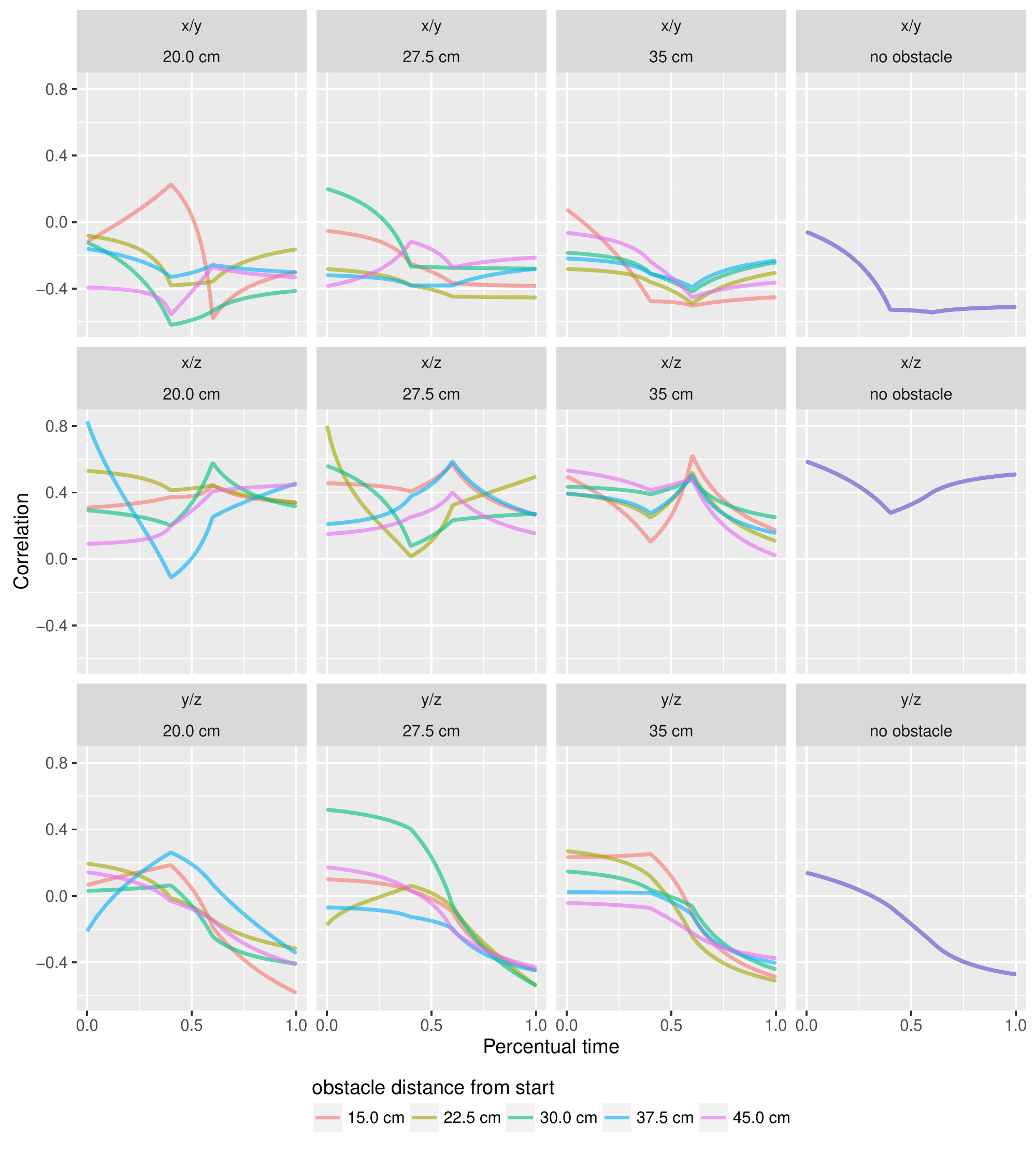} 
	\caption{Correlation functions over time as estimated by the proposed model in all 16 experimental set-ups.}
	\label{move-res-corr1} 
\end{figure}
Of particular interest is the correlation for the three axes (i.e. $x/y$, $x/z$ and $y/z$) and how it varies over time as seen in Figure~\ref{move-res-corr1}. From the results, it is clear that the correlations  vary over time, which Figure~\ref{3d-cov1} also illustrates. 
The variation of correlation with respect to time is moderate for the $x/y$- and $x/z$-correlations, but for the $y/z$-correlations there is a clear trend for all experimental set-ups that the correlation goes from positive values to negative values. This is a surprising and perhaps unexpected feature since all experimental set-ups are symmetric in the $y$-coordinate. 
A plausible explanation is that lifting a centrally placed object with the right hand is generally associated with moving that hand to the right (in our set-up, a positive $y$-value). When the object is raised  we observe a positive correlation in the $y/z$-plane (faster initial movement timing amplifies the effect), and when the object is lowered again we observe corresponding negative correlation.

\paragraph{Classification}\label{sec:classification}
To objectively compare different models, one can fit the models to a subset of the samples and compare their fits in terms of their classification accuracies of participant on the remaining data. That is, for a given functional sample that was not used to fit the model, we wish to determine which of the participants performed the movement. The primary objective of such an exercise is to compare similar generative models, but not as such to get the highest possible classification accuracy---a higher score could probably be achievable by standard machine learning methods that would reveal little about the structure of the problem. A similar classification-based approach was used to evaluate the hierarchical ``pavpop'' model described in \cite{RaketGrimme2016}, which was applied to the 1-dimensional acceleration magnitude profiles of the 3-dimensional arm movement data set.

The present classification was done in a chronological 5-fold cross-validation set-up (first fold consisted of the two first repetitions for each person, second fold of the third and fourth and so forth). Different models were fitted on the five training sets, each leaving out one of the folds (test set). For each test set, the samples were classified using the model estimates from the corresponding training set. The classification accuracy was then computed as the average classification accuracy across the five folds for each experiment. 

In the following, the proposed method is denoted by SIMM (Simultaneous Inference for Misaligned Multivariate curves). The following models were used in the comparison:

\begin{description}
	\item[Nearest centroid (NC)] The centroids for each person were estimated as the pointwise means in the training set. The classification was done using minimal Euclidean distance to the estimated centroid (using linear interpolation). 
	\item[Nearest centroid weighted (NC-W)] The centroids were computed similarly to the NC method, but the classification was done using a distance with weighted coordinates, the weights for the $x$-, $y$- and $z$-coordinates were $0.1/0.7/0.2$.
	\item[Fisher-Rao $L^2$ (FR-$L^2$)] Pointwise template functions were estimated using group-wise elastic function alignment and PCA extraction for modeling amplitude variation \citep{tucker2013generative, fdasrvf}. The standard setting of using 3 principal components was used. The elastic curve approach for functional data is widely considered the state-of-the-art framework for handling misaligned functional data \citep{marron2015functional}.  The template functions were estimated separately for each of the three value coordinates of the trajectories. Classification was done using minimal Euclidean distance to the estimated template functions. 
	\item[Fisher-Rao elastic (FR$_{\text{E}}$)] Template functions were estimated similarly to FR-$L^2$, but classification was done using an elastic distance that both measures coordinate-wise distances as a sum of phase \citep[Section 3.1]{tucker2013generative} and amplitude directions  \citep[Definition 1]{tucker2013generative}. The weighting between phase and amplitude distances was $0.16/0.84$.
	\item[Fisher-Rao elastic weighted (FR$_{\text{E}}$-W)] Template functions and classification was done similarly to FR$_{\text{E}}$, except that we include a weighting of the three elastic distances corresponding to each value coordinate. The weighting between phase and amplitude distances was $0.14/0.86$ and the weights for the $x$-, $y$- and $z$-components of the elastic distance were $0.3/0.2/0.5$.
	\item[Elastic curve metric (EM)]
		Multivariate elastic distance  between curves is defined as geodesic distance on $L^2([0,1]; \R^3)/\Gamma$, where $\Gamma$ is the closure of the set of positive diffeomorphisms on $[0,1]$. In the quotient space $L^2([0,1]; \R^3)/\Gamma$, all temporal features are removed and comparison of curves is done using only their image in $\R^3$, but in a way that is consistent with reparametrizations of the original curves \citep{srivastava_bog}. Templates were estimated as the pointwise averages of samples aligned to the Karcher mean in  $L^2([0,1]; \R^3)/\Gamma$ computed using the \texttt{fdasrvf} R-package \citep{fdasrvf}.
		Classification was done using a weighted sum of multivariate elastic distance and phase distance (defined as for the FR$_{\text{E}}$ method). The weighting between elastic and phase distances was $0.24/0.76$.
	\item[SIMM] The person-specific templates are estimated using the proposed model with a diagonal cross-covariance structure (i.e. no cross-covariance). Classification is done using nearest posterior distance under the maximum likelihood estimates as a function of the unknown sample.
	\item[SIMM-CC] Estimation and classification are done similarly to the SIMM method, but using the full dynamic cross-covariance structure described in the previous sections.
\end{description}

All weights described in the above methods were chosen by cross-validation on the accuracies for the three experimental set-ups with $d=30.0$ cm. The grids used for determining the parameters are given in the supplementary material.

The classification accuracies are available in Table~\ref{table:combined}. If we first consider the NC-type methods that do not model any warping effect, we see a marked increase in accuracy when weighting the different coordinates in the classification, and thus emulating a constant diagonal cross-covariance structure. If we consider the basic elastic model FR-$L^2$ based on the Fisher-Rao metric, we see similar results to the simple NC model, even though the  FR-$L^2$ method also accounts for a warping effect when estimating the template. When classifying using an elastic distance, as was done in FR$_{\text{E}}$, we see a great increase in classification accuracy. The phase distance contributes considerably to these improvements. When only considering elastic amplitude distance (i.e. weighting phase/amplitude distances 0/1) the average classification accuracy is 0.576. Taking the deformation distance into account in the classification, and thus paying a price for warping the templates, we see a great increase in classification accuracy. The heuristic idea of having to pay a price for large warps in many ways emulates the proposed idea of modeling the warping functions as random effects. Finally, the FR$_{\text{E}}$-W method includes a weighting of the combined phase and amplitude distances across the $x$-, $y$- and $z$-coordinates of the observed trajectories, which again increases the accuracy. 

The elastic metric has many similarities with the Fisher-Rao metric, but is multivariate in nature. The EM method has higher accuracies than the similar FR$_{\text{E}}$ and FR$_{\text{E}}$-W methods. Exploratory comparison of results suggested that this was caused by more appropriate warping across all coordinates leading to both better estimates of templates and in turn more accurate phase distances.

The SIMM model is the proposed model described above, but without a dynamic cross-correlation structure. Instead we have three scale parameters that describe the weighting of the marginal variances in the three value coordinates. The model is thus both comparable to FR$_{\text{E}}$-W and EM, both of which are outperformed in terms of accuracy. It is important to note that while FR$_{\text{E}}$-W and EM required cross-validation on a subset of the test data to estimate the parameters, the SIMM model estimates all variance parameters used in the weighting of the different aspect of the movement from the training data. The final model, SIMM-CC, includes a full dynamic cross-covariance structure. Even though one could anticipate that this model was much more prone to overfitting to the training data (the model includes 27 free amplitude variance parameters compared to the 6 parameters of the SIMM model), we see a slight increase in accuracy of the method. 
	We remark that the EM, SIMM and SIMM-CC methods, which make a joint warp of the three spatial coordinates, had the best accuracies among the methods in consideration. This strongly supports the idea of modeling multivariate signals with a joint warping of all value coordinates.

\begin{table}[!htp]
	\centering
	\resizebox{\textwidth}{!}{
		\begin{tabular}{cc|cccccccccccc}
			\hline
			$d$& obstacle & NC & NC-W & FR-$L^2$  & FR$_{\text{E}}$ & FR$_{\text{E}}$-W & EM & SIMM  & SIMM-CC\\
			\hline\hline
			& \emph{S} &0.62 & 0.71 & 0.58 &  0.77 & 0.79 & 0.77 & 0.80 & \textbf{0.85} \\ 
			15.0 cm & \emph{M} & 0.60 & 0.63 & 0.62 & 0.64 & 0.68 & 0.77 & 0.80 & \textbf{0.83}\\
			& \emph{T} & 0.52 & 0.57 &0.54 & 0.58 & 0.58 & 0.77 & \textbf{0.84} &  {0.81}\\
			\hline
			& \emph{S} & 0.51 & 0.58 & 0.50 & 0.68 & 0.66 & \textbf{0.77}  & 0.69   & \textbf{0.77}\\ 
			22.5 cm & \emph{M} & 0.52 & 0.64 & 0.56 & 0.62 & {0.73} & 0.70 & \textbf{0.75} &  {0.72}\\
			& \emph{T} & 0.50 & 0.62 & 0.49 & 0.64 & 0.73 & 0.73  & 0.74 &  \textbf{0.79}\\ 
			\hline
			& \emph{S} & 0.53 & \emph{0.59} & 0.53 & \emph{0.69} & \emph{0.72} & \emph{\textbf{0.76}} & 0.70 & {\textbf{0.76}} \\ 
			30.0 cm  & \emph{M} & 0.45 & \emph{0.47} & 0.48 & \emph{0.65} & \emph{0.68} & \emph{0.70} & \textbf{0.79} & {0.75}\\ 
			& \emph{T} & 0.58 & \emph{0.63} & 0.56 & \emph{0.65} & \emph{0.73} & \emph{0.78} & \textbf{0.86} & {0.83}\\ 
			\hline
			& \emph{S} & 0.51 & 0.55 & 0.52 & 0.67 & 0.72 & 0.70 & \textbf{0.77} & {0.76}\\
			37.5 cm  & \emph{M} & 0.45 & 0.50 & 0.43 & 0.68 & 0.65 & 0.69 & \textbf{0.68} & \textbf{0.68} \\
			& \emph{T} &  0.50 & 0.53 & 0.54 & 0.67 & 0.73 & 0.72  & \textbf{0.80} & \textbf{0.80}\\
			\hline
			& \emph{S} & 0.49 & 0.54 & 0.51 & 0.66 & 0.71 & 0.75 & 0.69 & \textbf{0.76}\\ 
			45.0 cm  & \emph{M} & 0.48 & 0.53 & 0.44 & 0.66 & 0.70 & 0.71 & \textbf{0.78} & {0.73}\\ 
			& \emph{T} & 0.50 & 0.54 & 0.50 & 0.71 & 0.75 & 0.74 & {0.82} & \textbf{0.83}\\ 
			\hline
			NA  & - & 0.48 & 0.56 & 0.52 & 0.68 & {0.72} & \textbf{0.80} & 0.64 & 0.70\\
			\hline\hline
			average & & 0.515 & 0.574 &  0.520 & 0.666 & 0.705 & 0.741 & 0.761 &  \bf{0.773}\\
			\hline
		\end{tabular} 
	} 
	
	\caption{Classification accuracies of various methods. {\bf Bold} indicates best result(s), \emph{italic} indicates that the given experiments were used for training.}\label{table:combined}
\end{table}

\section{Discussion}
In this paper we have proposed a new class of models for simultaneous inference for misaligned multivariate functional data. We fitted these types of models to three different data sets and applied it in one classification scenario. 

The idea behind the approach is to simultaneously model the predominant effects in functional data sets, misalignment and amplitude variation, as random effects. The simultaneous modeling allows separation of these effects in a data-driven manner, namely by maximum likelihood estimation. In particular, we saw that this separation resulted in nicely behaving warping functions that did not seem to over-align the functional samples. 

The models enable 
estimation of dynamic correlation functions between the individual coordinates of the amplitude variation. We demonstrated that one can achieve superior fits and better classification using the parametric construction from Proposition~\ref{prop-dyn-cov}, even when the number of free parameters is high relative to the number of functional samples. By fitting the model to two large functional data sets related to human movement, we also demonstrated the computational feasibility of maximum likelihood inference with such models.

The proposed parametric model class for dynamic covariance structures is very general, but other modeling approaches could be better suited in some situations.
For example, instead of using a fixed number of parameters to describe each marginal variance and cross-covariance function, one would often prefer to do this in a data-driven manner. One possibility could be to model the multivariate amplitude covariance function using a multivariate functional factor analysis model, for example a multivariate extension of the rank reduced model of \cite{james2000principal}, where the number of parameters describing the covariance is fixed, and the covariance is described in terms of functional principal components. However, such amplitude effects cannot be effectively fitted using conventional optimizers for the likelihood, and would require development of specialized efficient fitting methods (e.g. generalizing the methods of \citeauthor{peng2012geometric} \citeyear{peng2012geometric}).  
Another relevant approach would be simultaneous warping of fixed effects and amplitude variation, and one could also consider extending the domain of feasible warping functions by modelling the latent warp variables $w$ as more general functional objects (e.g. stochastic processes) instead of elements belonging to $\R^{m_{\bw}}$ for some $m_{\bw}$. We will leave these extensions as future work.

\bigskip
\begin{center}
	{\large\bf SUPPLEMENTARY MATERIAL}
\end{center}

\subsubsection*{Cross-validation grids} The cross-validation used to determine the parameters of the methods NC-W, FR$_{\text{E}}$, FR$_{\text{E}}$-W and EM in Section~\ref{sec:classification} were given as follows. The possible weights between the three value coordinates were $\{\bw\in \R^3 \,:\, w_i \in \{0, 0.1, \dots, 1\}, w_1+w_2+w_3=1\}$ and the possible weights between amplitude and phase distance were $\{\bw\in \R^2 \,:\, w_i \in \{0, 0.02, \dots, 1\}, w_1+w_2=1\}$. NC-W only uses weighting between value coordinates and FR$_{\text{E}}$ and EM only use weighting between the amplitude and phase distance. \\
For the SIMM-CC model we explored adding more than three knots to the warp model ($m_{\bw} = 3,4,5$), but $m_{\bw} = 3$ gave the best cross-validation score.

\subsubsection*{Covariance functions}

Below we list the covariance functions that are used in the three data examples. 

Schur's theorem states that the pointwise product of covariance functions yields a valid covariance function \citep{schur1911}. This property is used in the arm movement example.

\begin{description}
	\item[Brownian bridge] The covariance function for the Brownian bridge defined on the temporal domain $[0,1]$ is given by: 
	\begin{equation} \label{cov:bridge}
	f_\text{bridge}(s,t) = \tau^2 \min(s,t) \cdot (1- \max(s,t)) = \tau^2(\min(s,t) - st), \quad s,t \in [0,1].
	\end{equation}
	where $\tau > 0$ is a scale parameter.
	\item[Brownian motion] The covariance function for the Brownian motion defined on the domain $[0,\infty)$ is given by:
	\begin{equation} \label{cov:motion}
	f_\text{motion}(s,t) = \tau^2 \min(s,t), \quad s,t \geq 0.
	\end{equation}
	where $\tau > 0 $ is a scale parameter.
	\item[Mixing stationary and bridge covariances] 
	The combination of a stationary and bridge covariance with mixtures $a$ and $b$ is given by
	\begin{equation*} 
	f_{\text{mixture}(a,b)}(s,t) = a + b \cdot  (\min(s,t) - st)
	\end{equation*}
	In our analysis the parameter $b$ is redundant, so we use 
	\begin{equation} \label{cov:mixture}
	f_{\text{mixture}(a)}(s,t) = a + \min(s,t) - st
	\end{equation}	
	
	Note that the bridge covariance is not the same construction as when conditioning a stochastic process $X$ on its endpoint value. 
	\item[Mat\'{e}rn covariance function]  The covariance function for the Mat\'{e}rn covariance with smoothness parameter $\alpha$ and range parameter $\kappa$ is given by: 
	\begin{equation} \label{cov:matern}
	f_\text{Mat\'{e}rn($\alpha$,$\kappa$)}(s,t) = \frac{2^{1-\alpha}}{\Gamma(\alpha)} (|s-t|/\kappa )^\alpha K_\alpha(|s-t|/\kappa), \quad s,t\in \R.
	\end{equation}
	Here $K_\alpha$ is the modified Bessel function of the second kind.  A Gaussian process with Mat\'{e}rn covariance is stationary, and conversely any stationary continuous Gaussian process with mean zero has a covariance function that up to scale is given by a Mat\'{e}rn covariance function \citep{RasmussenWilliams}.
\end{description}

\newpage

\subsubsection*{Parameter estimates for Arm movement data}

\begin{table}[!h]
	\centering
	\begin{tabular}{cc | ccccc}
		\hline
		$d$& obstacle & $\sigma$ & $\alpha$ & $\kappa$ & $a$ & $ \sigma \tau$  \\
		\hline\hline
		& \emph{S} &0.0012 & 1.432 & 0.157 & 19.56 & 0.0519 \\
		15.0 cm & \emph{M} & 0.0012 &1.749 & 0.120 & 24.60 & 0.0525 \\
		& \emph{T} & 0.0013 & 1.627 & 0.124 & 22.54 &  0.0502\\
		\hline
		& \emph{S} & 0.0013 & 1.788 & 0.128 & 25.13 & 0.0531  \\
		22.5 cm & \emph{M} &0.0011 & 1.638 & 0.139 & 58.20 & 0.1177\\
		& \emph{T} & 0.0012 & 1.679 & 0.121 & 26.37 & 0.0528 \\
		\hline
		& \emph{S} &0.0012 & 1.773 & 0.121 & 20.96 & 0.0549 \\ 
		30.0 cm  & \emph{M} &0.0014 & 1.663 & 0.139 & 21.31 & 0.0518 \\ 
		& \emph{T} & 0.0012 & 1.687 & 0.128 & 26.63 & 0.0643 \\ 
		\hline
		& \emph{S} &  0.0012 & 1.481 & 0.155 & 17.69 & 0.0622 \\ 
		37.5 cm  & \emph{M} &0.0013 & 1.658 & 0.125 & 19.80 & 0.0596 \\ 
		& \emph{T} &  0.0010 & 1.633 & 0.121 & 34.29 & 0.0563 \\ 
		\hline
		& \emph{S} & 0.0013 & 1.761 & 0.123 & 19.10 & 0.0504 \\ 
		45.0 cm  & \emph{M} & 0.0016 & 1.760 & 0.119 & 13.09 & 0.0668 \\ 
		& \emph{T} &0.0010 & 1.670 & 0.121 & 37.46 & 0.0548 \\ 
		\hline
		NA  & - & 0.0009 & 1.786 & 0.142 & 47.04 & 0.0561 \\ 
		\hline
	\end{tabular}
	\caption{Parameter estimates for the arm movement data. 
	}\label{table:suppl_param}
\end{table}

\subsection*{EM algorithm for the spline coefficients in the linearized model}

First note that by assumption the mean curves $\btheta$ are the same, expect for warping, for trajectories belonging to the same subject groups and are independent of other subject groups.
Thus, in order to simplify notation and ease argumentation, we will assume that all trajectories belong to the same subject group.

Let $f = \{f_k \}$ be the spline base function for $\btheta$ and let $\bc$ be the spline coefficients, i.e. $\btheta(t) = f(t) \cdot \bc$. Consider the linearized model from Equation \eqref{lin-mod}:
\begin{equation*}
\vec{\by}_n \approx \vec{\bgamma}_{\bw_n^0} +  Z_n(\bw_n - \bw_n^0) + \vec{\bx}_n + \vec{\bepsilon}_n, \quad n=1,\dotsc,N
\end{equation*}
with log-likelihood
\begin{equation*} 
\sum_{n=1}^N \bigg( q m_n \log \sigma^2 + \log\det V_n \\
+ \sigma^{-2}(\vec{\by}_n -  \vec{\bgamma}_{\bw_n^0} + Z_n \bw_n^0)^\top V_n^{-1} (\vec{y}_n -  \vec{\bgamma}_{\bw_n^0} +  Z_n \bw_n^0 ) \bigg).
\end{equation*}
For the remainder we assume that $\bw^0_n = \{ \bw^0_{nl} \}_{l=1}^{m_w}$ and all variance parameters $(\mathcal{S}, C, \sigma^2)$ are fixed, and that we have a current estimate of the spline coefficients, $\bc_0$. The conditional expectation and variance of $\bw_n$ given the observations $\by$ under the current parameters will be denoted by $\bar{\bw}_n = \{ \bar{\bw}_{nl} \}_{l=1}^{m_w} \in \R^{m_w}$ and $\bar{\bar{\bw}}_n = \{\bar{\bar{\bw}}_{n l_1 l_2} \}_{l_1,l_2=1}^{m_w} \in \R^{m_w \times m_w}$, respectively. Using this notation the conditional log-likelihood of $\vec{\by}_n$ given $\bw_n$ is
\begin{equation*}
l_{\vec{\by}_n|\bw_n} = (\vec{\by}_n - \vec{\bgamma}_{\bw_n^0}  - Z_n(\bw_n-\bw^0_n) )^\top S_n^{-1}  (\vec{\by}_n - \vec{\bgamma}_{\bw_n^0} - Z_n(\bw_n-\bw^0_n) ) + \log\det S_n.
\end{equation*}
The term $\log\det S_n$ does not influence the estimation of $\bc$, and hence it will be removed in the following. The conditional expectation $E[l_{\vec{\by}_n|\bw_n} | \vec{\by}_n]$ given the observation hence equals
\begin{equation} \label{exp-l-y}
(\vec{\by}_n - \vec{\bgamma}_{\bw_n^0} - Z(\bar{\bw}_n-\bw^0_n) )^\top S_n^{-1} (\vec{\by}_n - \vec{\bgamma}_{\bw_n^0} - Z(\bar{\bw}_n-\bw^0_n)) + \tr(S_n^{-1} Z_n \bar{\bar{\bw}}_n Z_n^\top ).
\end{equation}
Defining $R_n = f(v(t_k, \bw^0_n))$ and $R_{nl} = \partial_t f(v(t_k, \bw^0_n)) \partial_{\bw_l} v(t_k, \bw_n^0)$ for $l = 1,\dotsc,m_w$ we have that $Z_n = \{ 
R_{nl} \cdot \bc \}_{l=1}^{m_w}$ and thus $Z_n \bw_n = (\sum_{l=1}^{m_w} \bw_{nl}R_{nl})\cdot \bc $. Using this the trace from \eqref{exp-l-y} can be expanded as a double sum
\begin{equation*}
\tr(S_n^{-1} Z \bar{\bar{\bw}}_n Z^\top) = \sum_{l_1,l_2 = 1}^{m_w} \bar{\bar{\bw}}_{n l_1 l_2} \tr \left( S_n^{-1} R_{nl_1} \bc\bc^\top R_{nl_2}^\top \right) .
\end{equation*}

Calculating the gradient of \eqref{exp-l-y} now gives that $\nabla_{\bc} E[l_{\vec{\by}_n|\bw_n} | \vec{\by}_n]$ is proportional to
\begin{equation*}
-K_n^\top S_n^{-1}  (\vec{\by}_n - K_n \bc) + 
\sum_{l_1,l_2 = 1}^{m_w} \bar{\bar{\bw}}_{n l_1 l_2}  R_{nl_2}^\top S_n^{-1} R_{nl_1} \bc,
\end{equation*}
where $K_n = R_n +  \sum_{l=1}^{m_w} (\bar{\bw}_{nl} - \bw^0_{nl}) R_{nl}$. From this it follows that the M-step of the EM algorithm for the spline coefficients $\bc$ is given by
\begin{equation*}
\bc_{new} = \left[\sum_{n=1}^N K_n^\top S_n^{-1}K_n  + \sum_{l_1,l_2 = 1}^{m_w}  \bar{\bar{\bw}}_{n l_1 l_2}  R_{nl_1} S_n^{-1} R_{nl_2}^\top \right]^{-1} \sum_{n=1}^N K_n^\top S_n^{-1} y_n.
\end{equation*}

	\bibliographystyle{agsm}
	
	\bibliography{bib}
\end{document}